\documentclass[runningheads]{llncs}[2018/03/10]
\usepackage[english]{babel}
\usepackage[T1]{fontenc}
\usepackage[utf8]{inputenc} 

\usepackage[pdftex,usenames,dvipsnames,svgnames]{xcolor}

\colorlet{LegColor}{black}

\usepackage{url}
\makeatletter
\g@addto@macro{\UrlBreaks}{\UrlOrds}
\makeatother

\usepackage{hyperref}
\hypersetup{colorlinks=true,           allcolors=blue,            pdfstartview=Fit,          breaklinks=true,
pdfauthor={V. Danos, T. Heindel, R. Honorato-Zimmer and S. Stucki},
  pdftitle={Rate Equations for Graphs}}
\usepackage[all]{hypcap}

\usepackage{xifthen}

\usepackage{tabularx}
\usepackage{subcaption}

\RequirePackage[babel]{microtype}

\DisableLigatures{encoding = T1, family = tt* }

\RequirePackage{amsmath}    \RequirePackage{amsfonts}   \RequirePackage{amssymb}    \RequirePackage{mathtools}

\RequirePackage{dsfont}

\RequirePackage{stmaryrd}

\RequirePackage{tikz}
\usetikzlibrary{arrows,fit,matrix,positioning,calc}

     \renewcommand*{\exp}{\operatorname{\mathbf{E}}}          

     \DeclareMathOperator{\EE}{\mathbb{E}}

  \newcommand{\NN}{\mathbb{N}}                                 \newcommand{\RR}{\mathbb{R}}

\newcommand{\st}{\mid}

\newcommand{\set}[1]{\left\{ #1 \right\}}

\newcommand{\setof}[2]{\set{#1 \st #2}}

\newcommand{\tuple}[1]{\left( #1 \right)}

\newcommand{\abs}[1]{\left\lvert #1 \right\rvert}

    \DeclareMathOperator*{\cod}{cod}        

\DeclareMathOperator*{\comp}{\circ}

\newcommand{\ccat}[1]{\ensuremath{\boldsymbol{\mathbf{#1}}}}

\newcommand{\cat}[1]{\ensuremath{\mathcal{#1}}}

       \newcommand{\Grph}{\ccat{Grph}}

\newcommand{\Sub}{\ccat{Sub}}

\DeclareMathOperator{\id}{id}     

\newcommand{\ifempty}[3]{\ifx\\#1\\ #2\else #3\fi}

\newcommand{\toby}[1]{\to[#1]}

\renewcommand{\to}[1][]{\ifempty{#1}{\rightarrow}{\xrightarrow{\, #1 \,}}}
\renewcommand{\gets}[1][]{\ifempty{#1}{\leftarrow}{\xleftarrow{\, #1 \,}}}

\newcommand{\ntrans}[1]{\stackrel{\bullet}{\rightarrow}}

\newcommand{\pair}[2]{\langle #1, #2 \rangle}

\newcommand{\avg}[1]{\langle #1 \rangle}

\usepackage{tikz}
\usepackage{tikz-cd}
\usetikzlibrary{arrows,backgrounds,fit,matrix}

\usepackage{xspace}

\def\al{\alpha}
\def\ba{\beta}

   \renewcommand*{\exp}{\EE}

\newcommand{\ddt}{\tfrac{d}{dt}}

\renewcommand{\avg}[1]{\exp_p #1}
\newcommand{\dotavg}[1]{\frac{d}{dt}\avg{#1}}

\newcommand{\obs}[1]{\left[#1\right]}

\newcommand{\tsum}{\textstyle\sum}

\newcommand{\GrphP}{\ccat{\Grph_*}}

\newcommand{\Gs}{\cat{G}}             \newcommand{\GsP}{\cat{G_*}}                 

\newcommand{\mchs}[2]{\left[{#1}, {#2}\right]}
\newcommand{\nmchs}[2]{\abs{\mchs{#1}{#2}}}

\newcommand{\rwto}[1][]{\ifempty{#1}{\Rightarrow}{\Rightarrow_{#1}}}

\newcommand{\rto}[1][]{\ifempty{#1}{\rightharpoonup}{\xrightharpoonup{\, #1 \;}}}

\newcommand{\dg}[1]{#1^\dag}

\newcommand{\MGs}[3][]{#2 *_{#1} #3}
\newcommand{\MGCls}[2]{\MGs[\simeq]{#1}{#2}}
\newcommand{\LMGs}[2]{\MGs[L]{#1}{#2}} \newcommand{\RMGs}[2]{\MGs[R]{#1}{#2}} 

\newcommand{\rset}{\mathcal{R}}

         \newcommand{\ie}{i.e.\@\xspace}
     \newcommand{\eg}{e.g.\@\xspace}

\newcommand{\cf}{cf.\@\xspace}

\newcommand{\pref}[2]{\hyperref[#2]{#1\ref*{#2}}}

\newcommand{\Sec}[2][\S]{\pref{#1}{sec:#2}}
\newcommand{\App}[2][App.]{\pref{#1~}{app:#2}}
\newcommand{\Fig}[2][Fig.]{\pref{#1~}{fig:#2}}

\newcommand{\Def}[2][Def.]{\pref{#1~}{def:#2}}
\newcommand{\Lem}[2][Lemma]{\pref{#1~}{lem:#2}}

\newcommand{\eqnlabel}[1]{\refstepcounter{equation}\textup{(\theequation)}\label{#1}}

\newenvironment{diags}[1][0pt]{\begin{center}\vspace{#1}\def\diagsspaceafter{#1}}{\vspace{\diagsspaceafter}\end{center}}

\newenvironment{narrow}[2][0pt]{\list{}{\setlength{\listparindent}{\parindent}\setlength{\labelwidth}{0pt}\setlength{\topsep}{0pt}\setlength{\partopsep}{0pt}\setlength{\parsep}{\parskip}\setlength{\leftmargin}{#1}\setlength{\rightmargin}{#2}}\item\relax\ignorespaces}{\endlist}

\newcommand{\rightoverlay}[3][\rightmargin]{\newbox\roverlaybox\setbox\roverlaybox\hbox{#3}\begin{tikzpicture}\path[use as bounding box] (0.5\wd\roverlaybox - #1, -0.5*#2);\pgftext{\usebox{\roverlaybox}}\end{tikzpicture}}

\tikzstyle{mono}=[commutative diagrams/tail]
\tikzstyle{epi}=[commutative diagrams/two heads]

\tikzstyle{eq}=[commutative diagrams/equal]

\tikzstyle{nat}=[commutative diagrams/Rightarrow]

\tikzstyle{mtch}=[mono]
\tikzstyle{rule}=[commutative diagrams/rightharpoonup]
\tikzstyle{elur}=[commutative diagrams/rightharpoondown]
\tikzstyle{RULE}=[elur] \tikzstyle{inj}=[mono]

\tikzstyle{birule}=[cm right to-cm left to]

\tikzstyle{grphdiag-bg}=[rounded corners, fill=white!85!black]

\tikzstyle{grphdiag}=[>=stealth, framed, thick,background rectangle/.style=grphdiag-bg]

\tikzstyle{grphnode}=[rectangle,grphdiag-bg]

\tikzstyle{ingrphdiag}=[>=stealth, semithick,n/.append style={semithick, minimum size=4pt}]

\newcommand{\arrsn}[3][.1]{\arr[#1]{#2.south}{#3.north}}
\newcommand{\arr}[3][.1]{\draw[->] ($(#2)!  #1!(#3)$) --
            ($(#2)!1-#1!(#3)$);}

\tikzstyle{n}=[circle, draw, thick, minimum size=6pt,inner sep=0pt] \tikzstyle{b}=[n, fill=DarkGray] \tikzstyle{c}=[n, fill=white] \tikzstyle{p}=[n, draw=none] \tikzstyle{g}=[n, draw=DarkGreen, fill=white!80!DarkGreen] \tikzstyle{r}=[n, draw=BrickRed, fill=white!80!BrickRed] 

\newcommand{\n}[3][n]{\node[#1] (#2) at (#3) {}}

\newcommand{\conc}[1]{\, \mathbin{\tikz[grphdiag, x=.4cm, y=.4cm, baseline=3.5, semithick,n/.append style={minimum size=4pt}]{#1}} \,}
\newcommand{\bigconc}[1]{\: \mathbin{\tikz[grphdiag, x=.7cm, y=.6cm, baseline=6.5]{#1}} \:}

\newcommand{\birulearrow}[2]{\hspace{1ex}\tikz{\draw[rule] (1,   .1) -- node[above] {#1} ++( 1.2, 0);
  \draw[rule] (2.2, .0) -- node[below] {#2} ++(-1.2, 0);
}\hspace{1ex}}

\newcommand{\src}[1][]{\ensuremath{s_{#1}}}

\newcommand{\tgt}[1][]{\ensuremath{t_{#1}}}

\tikzset{commutative diagrams/row sep/normal=1.5em}
\tikzset{commutative diagrams/column sep/normal=1.6em}

\allowdisplaybreaks

\newcommand{\Example}{\Sec{intro}}

\begin{document}

\title{Rate Equations for Graphs\thanks{This is a preprint of a paper to be presented at the 18th
    International Conference on Computational Methods in Systems
    Biology (CMSB~2020).
The conference version will be published in Springer's LNCS/LNBI
    series.
RH-Z was supported by ANID FONDECYT/POSTDOCTORADO/No3200543.
  }}
\titlerunning{Rate Equations for Graphs}
\author{Vincent~Danos\inst{1} \and
  Tobias~Heindel\inst{2} \and
  Ricardo~Honorato-Zimmer\inst{3} \and
  Sandro~Stucki\inst{4} }
\authorrunning{V.~Danos et al.}
\institute{CNRS, ENS-PSL, INRIA,
  Paris, France, \email{vincent.danos@ens.fr} \and
  Institute of Commercial Information Technology and Quantitative Methods, \\
Technische Universit\"at Berlin, Germany, \email{heindel@tu-berlin.de} \and
Centro Interdisciplinario de Neurociencia de Valparaíso, \\
  Universidad de Valparaíso, Chile, \email{ricardo.honorato@cinv.cl} \and
  Dept.\ of Computer Science and Engineering,
  University of Gothenburg, Sweden,
  \email{sandro.stucki@gu.se}
}
\maketitle

\begin{abstract}
  In this paper, we combine ideas from two different scientific
  traditions:
1) graph transformation systems (GTSs) stemming from the theory of
  formal languages and concurrency, and
2) mean field approximations (MFAs), a collection of
  approximation techniques ubiquitous in the study of complex
  dynamics.
Using existing tools from algebraic graph rewriting, as well as new
  ones, we build a framework which generates rate equations for
  stochastic GTSs and from which one can derive MFAs of any order (no
  longer limited to the humanly computable).
The procedure for deriving rate equations and their approximations
  can be automated.
An implementation and example models are available online at
  \url{https://rhz.github.io/fragger}.
We apply our techniques and tools to derive an expression for the
  mean velocity of a two-legged walker protein on DNA.

   \end{abstract}

\section{Introduction}\label{sec:intro}

Mean field approximations (MFAs) are used in the study
of complex systems to obtain simplified and revealing
descriptions of their dynamics.
MFAs are used in many disparate contexts such as
Chemical Reaction Networks (CRNs) and their derivatives~\cite{vankampen,kurtz,bortolussi2013continuous},
walkers on bio-polymers~\cite{tasep,Stukalin2005},
models of epidemic spreading~\cite{gleeson},
and the evolution of social networks~\cite{durrett2012graph}.
These examples witness both the power and universality of MFA
techniques, and the fact that they are pursued in a seemingly
ad~hoc, case-by-case fashion.

The case of CRNs is particularly interesting
because they provide a human-readable,
declarative language for a common class of complex systems.
The stochastic semantics of a CRN is given by
a continuous-time Markov chain (CTMC)
which gives rise to the so-called \emph{master equation} (ME).
The ME is a system of differential equations describing the time
evolution of the probability of finding the CRN in any given state.
Various tools have been developed to automate
the generation and solution of the ME from a given CRN,
liberating modellers from the daunting task of working with
the ME directly (\eg~\cite{FagesS08sfm,ThomasMG12pone,LopezMBS13msb}).

Its high dimensionality often precludes exact solutions of the ME.
This is where MFA techniques become effective.
The generally countably infinite ME is replaced by
a finite system of differential equations,
called the \emph{rate equations} (RE)~\cite{vankampen,ramonemre},
which describe the time evolution of
the average occurrence count of individual species.
Here, we extend this idea to the case of graphs and,
in fact, the resulting framework subsumes all the examples
mentioned above (including CRNs).
The main finding is summarised in a single equation \eqref{eq:GREG}
which we call the \emph{generalised rate equation for graphs} (GREG).
In previous work, we have published a solution to this problem
for the subclass of reversible graph rewriting
systems~\cite{icfem,rc15}.
The solution presented here is valid for \emph{any} such system, reversible or not.
The added mathematical difficulty is substantial and concentrates
in the backward modularity \Lem{bmod}.
As in Ref.~\cite{rc15}, the somewhat informal approach of Ref.~\cite{icfem} is replaced with precise category-theoretical language with which the backward modularity Lemma finds a concise and natural formulation.

As the reader will notice, Equation~\eqref{eq:GREG}
is entirely combinatorial and can be readily implemented.
Our implementation can be played with at~\url{https://rhz.github.io/fragger}.
Its source can be found at~\url{https://github.com/rhz/fragger}.

\subsection{Two-legged DNA walker}\label{subsec:bimotor}

\newcommand*{\Ga}{\n[b]{b}{0,1};
  \n[c]{c1}{0,0};
  \n[c]{c2}{1,0};
  \draw (b) edge[bend right,->] (c1);
  \draw (b) edge[bend left,->,LegColor] (c1);
  \draw[->] (c1) -- (c2);}
\newcommand*{\bigGa}{\n[b]{b}{0,1};
  \n[c]{c1}{0,0};
  \n[c]{c2}{1,0};
  \draw (b.south west) edge[bend right,->] (c1.north west);
  \draw (b.south east) edge[bend left,->,LegColor] (c1.north east);
  \draw[->] (c1) -- (c2);}
\newcommand*{\Gb}{\n[b]{b}{.5,1};
  \n[c]{c1}{0,0};
  \n[c]{c2}{1,0};
  \draw[->] (b) -- (c1);
  \draw[->,LegColor] (b) -- (c2);
  \draw[->] (c1) -- (c2);}
\newcommand*{\Gc}{\n[b]{b}{1,1};
  \n[c]{c1}{0,0};
  \n[c]{c2}{1,0};
  \draw (b) edge[bend right,->] (c2);
  \draw (b) edge[bend left,->,LegColor] (c2);
  \draw[->] (c1) -- (c2);}
\newcommand*{\bigGc}{\n[b]{b}{1,1};
  \n[c]{c1}{0,0};
  \n[c]{c2}{1,0};
  \draw (b.south west) edge[bend right,->] (c2.north west);
  \draw (b.south east) edge[bend left,->,LegColor] (c2.north east);
  \draw[->] (c1) -- (c2);
}
\newcommand*{\Go}{\n[b]{b}{0,1};
  \n[c]{c1}{0,0};
  \draw (b) edge[bend right,->] (c1);
  \draw (b) edge[bend left,->] (c1);}

Let us start
with an example from biophysics~\cite{Stukalin2005}.
The model describes a protein complex walking on DNA.
The walker contains two special proteins -- the \emph{legs} --
each binding a different DNA strand.
The legs are able to move along the strands independently
but can be at most $m$ DNA segments apart.

Following Stukalin et~al.~\cite{Stukalin2005},
we are interested in computing the velocity at which
a two-legged walker moves on DNA with $m=1$.
In this case, and assuming the two legs are symmetric,
there are only two configurations a walker can be in:
either extended (E) or compressed (C).
Therefore all possible transitions can be compactly represented by
the four rules shown in \Fig{bimotor-rules},
where the grey node represents the walker
and white nodes are DNA segments.
The polarisation of the DNA double helix is represented
by the direction of the edge
that binds two consecutive DNA segments.
Rules are labelled by two subscripts:
the first tells us if the leg that changes position is moving
\emph{forward}~(F) or \emph{backward}~(B), while the second
states whether the rule extends or compresses the current configuration.

\def\enp{\mathbb E}

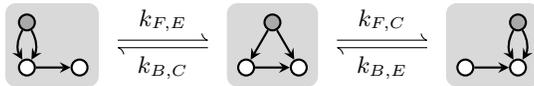
\begin{figure}[tb]
  \begin{center}
    \begin{tikzpicture}[grphdiag, x=.7cm, y=.6cm]
      \Ga \end{tikzpicture}
    \birulearrow{$k_{F,E}$}{$k_{B,C}$}
    \begin{tikzpicture}[grphdiag, x=.7cm, y=.6cm]
      \Gb
    \end{tikzpicture}
    \birulearrow{$k_{F,C}$}{$k_{B,E}$}
    \begin{tikzpicture}[grphdiag, x=.7cm, y=.6cm]
      \Gc \end{tikzpicture}
  \end{center}
  \vspace{-.4cm}
\caption{Stukalin model of a walking DNA bimotor.}\label{fig:bimotor-rules}
\end{figure}

The \emph{mean velocity} $V$ of a single walker in the
system can be computed from the rates at which they move
forward and backward and their expected number of occurrences
$\enp{\obs{G_i}}$, where $G_i$ is in either of the three
possible configurations depicted in Fig.~\ref{fig:bimotor-rules},
and $\obs{G_i}$ is short for $\obs{G_i}(X(t))$,
the integer-valued random variable that tracks
the number of occurrences of $G_i$
in the (random) state of the system $X(t)$ at time $t$.
We call any real- or integer-valued function on $X(t)$
an \emph{observable}.
\begin{equation*}
   V = \frac{1}{2} \left(
     k_{F,E} \enp{\obs{\conc{\Ga}}}
   + k_{F,C} \enp{\obs{\conc{\Gb}}}
   - k_{B,E} \enp{\obs{\conc{\Gc}}}
   - k_{B,C} \enp{\obs{\conc{\Gb}}} \right)
\end{equation*}
In the case there is only a single motor in the system,
the observables $\obs{G_i}$ are Bernoulli-distributed random variables,
and the expectations $\enp{\obs{G_i}}$ correspond to
the probabilities of finding the motor
in the configuration $G_i$ at any given time.
Thus by constructing the ODEs for these observables,
we can compute the mean velocity of a single motor in the system.
That is, we must compute the rate equations for these graphs.

Intuitively, to compute rate equations
we must find all ways in which the rules can
create or destroy an occurrence of an observable of interest.
When, and only when, a rule application and
an occurrence of the given observable overlap,
can this occurrence be created or destroyed.
A systematic inventory of all such overlaps
can be obtained by enumerating the so-called
\hyperref[def:mgs]{\emph{minimal gluings}} (MGs)
of the graph underlying the given observable and
the left- and right-hand sides of each rule in the system.
MGs show how two graphs can overlap (full definition in the next section).
Such an enumeration of MGs is shown in \Fig{bimotor-mgs},
where the two graphs used to compute the MGs are the extended walker motif
-- the middle graph in \Fig{bimotor-rules} -- and
the left-hand side of the forward-extension rule.
The MGs are related and partially ordered by graph morphisms between them.

In theory, since we are gluing with the left-hand side of a rule
each one of the MGs represents a configuration in which
the application of the rule might destroy
an occurrence of the observable.
However, if we suppose that walkers initially have two legs,
then 13 of the 21 MGs in~\Fig{bimotor-mgs}
are impossible to produce by the rules,
because no rule can create additional legs.
Therefore those configurations will never be reached
by the system and we can disregard them.
If we further suppose the DNA backbone to be simple
and non-branching, we eliminate three more gluings.
Finally, if there is only one motor,
the remaining four non-trivial gluings are eliminated.
In this way, invariants can considerably reduce
the number of gluings that have to be considered.
Removing terms corresponding to observables which,
under the assumptions above, are identically zero,
we get the following series of ODEs.
For readability, only a subset of the terms is shown,
and we write $G$ instead of the proper $\enp{\obs{G}}$ in ODEs.

\begin{alignat*}{3}
\frac{d}{dt}\conc{\Gb} & {}={} &
  & k_{F,E} \conc{\Ga} - k_{B,C} \conc{\Gb}
  - k_{F,C} \conc{\Gb} + k_{B,E} \conc{\Gc}
  \\
\frac{d}{dt}\conc{\Ga} & {}= -{} &
  & k_{F,E} \conc{\Ga} + k_{B,C} \conc{\Gb}
  + k_{F,C} \conc{\n[b]{b}{.5,1};
    \n[c]{c1}{0,0};
    \n[c]{c2}{1,0};
    \n[c]{c3}{2,0};
    \draw[->] (b) -- (c1);
    \draw[->,LegColor] (b) -- (c2);
    \draw[->] (c1) -- (c2);
    \draw[->] (c2) -- (c3);} - \ldots
\\
\frac{d}{dt}\conc{\n[b]{b}{.5,1};
    \n[c]{c1}{0,0};
    \n[c]{c2}{1,0};
    \n[c]{c3}{2,0};
    \draw[->] (b) -- (c1);
    \draw[->,LegColor] (b) -- (c2);
    \draw[->] (c1) -- (c2);
    \draw[->] (c2) -- (c3);} & ={} & & k_{F,E} \conc{\n[b]{b}{0,1};
    \n[c]{c1}{0,0};
    \n[c]{c2}{1,0};
    \n[c]{c3}{2,0};
    \draw (b) edge[bend right,->] (c1);
    \draw (b) edge[bend left,->,LegColor] (c1);
    \draw[->] (c1) -- (c2);
    \draw[->] (c2) -- (c3);} - k_{B,C} \conc{\n[b]{b}{.5,1};
    \n[c]{c1}{0,0};
    \n[c]{c2}{1,0};
    \n[c]{c3}{2,0};
    \draw[->] (b) -- (c1);
    \draw[->,LegColor] (b) -- (c2);
    \draw[->] (c1) -- (c2);
    \draw[->] (c2) -- (c3);} - k_{F,C} \conc{\n[b]{b}{.5,1};
    \n[c]{c1}{0,0};
    \n[c]{c2}{1,0};
    \n[c]{c3}{2,0};
    \draw[->] (b) -- (c1);
    \draw[->,LegColor] (b) -- (c2);
    \draw[->] (c1) -- (c2);
    \draw[->] (c2) -- (c3);} + \ldots \\
\frac{d}{dt}\conc{\n[b]{b}{0,1};
    \n[c]{c1}{0,0};
    \n[c]{c2}{1,0};
    \n[c]{c3}{2,0};
    \draw (b) edge[bend right,->] (c1);
    \draw (b) edge[bend left,->,LegColor] (c1);
    \draw[->] (c1) -- (c2);
    \draw[->] (c2) -- (c3);} & = -{} & & k_{F,E} \conc{\n[b]{b}{0,1};
    \n[c]{c1}{0,0};
    \n[c]{c2}{1,0};
    \n[c]{c3}{2,0};
    \draw (b) edge[bend right,->] (c1);
    \draw (b) edge[bend left,->,LegColor] (c1);
    \draw[->] (c1) -- (c2);
    \draw[->] (c2) -- (c3);} + k_{B,C} \conc{\n[b]{b}{.5,1};
    \n[c]{c1}{0,0};
    \n[c]{c2}{1,0};
    \n[c]{c3}{2,0};
    \draw[->] (b) -- (c1);
    \draw[->,LegColor] (b) -- (c2);
    \draw[->] (c1) -- (c2);
    \draw[->] (c2) -- (c3);} + k_{F,C} \conc{\n[b]{b}{.5,1};
    \n[c]{c1}{0,0};
    \n[c]{c2}{1,0};
    \n[c]{c3}{2,0};
    \n[c]{c4}{3,0};
    \draw[->] (b) -- (c1);
    \draw[->,LegColor] (b) -- (c2);
    \draw[->] (c1) -- (c2);
    \draw[->] (c2) -- (c3);
    \draw[->] (c3) -- (c4);} - \ldots \\
\frac{d}{dt}\conc{\n[b]{b}{.5,1};
    \n[c]{c1}{0,0};
    \n[c]{c2}{1,0};
    \n[c]{c3}{2,0};
    \n[c]{c4}{3,0};
    \draw[->] (b) -- (c1);
    \draw[->,LegColor] (b) -- (c2);
    \draw[->] (c1) -- (c2);
    \draw[->] (c2) -- (c3);
    \draw[->] (c3) -- (c4);} & ={} & & \ldots
\end{alignat*}

Notice how
only graphs with extra white nodes to the right are obtained
when computing the ODE for the left graph in \Fig{bimotor-rules}.
The opposite is true for the right graph in \Fig{bimotor-rules}.
This infinite expansion can be further simplified
if we assume the DNA chain to be infinite or circular.
In this case we can avoid boundary conditions and replace the left-
and right-hand observables below by the simpler middle observable:
\begin{equation*}
  \enp{\obs{\bigconc{\Ga}}} =
  \enp{\obs{\bigconc{\Go}}} =
  \enp{\obs{\bigconc{\Gc}}}
\end{equation*}
The infinite expansion above now
boils down to a simple finite ODE system.
\begin{alignat*}{3}
\frac{d}{dt}\conc{\Gb} & ={} &
  & k_{F,E} \conc{\Go} - k_{B,C} \conc{\Gb}
  - k_{F,C} \conc{\Gb} + k_{B,E} \conc{\Go} \\
\frac{d}{dt}\conc{\Go} & = -{} &
  & k_{F,E} \conc{\Go} + k_{B,C} \conc{\Gb}
  + k_{F,C} \conc{\Gb} - k_{B,E} \conc{\Go}
\end{alignat*}
From the above ODEs and assumptions,
we get the steady state equation.
\begin{align*}
(k_{F,E} + k_{B,E}) \enp{\obs{\conc{\Go}}} & {}=
(k_{F,C} + k_{B,C}) \enp{\obs{\conc{\Gb}}}
\end{align*}
Since we have only one motor,
\[ \enp{\obs{\conc{\Go}}} + \enp{\obs{\conc{\Gb}}} = 1 \]
Using this, we can derive
the steady state value for the mean velocity:
\begin{align*}
  V & {}= \frac{1}{2} \biggl(
(k_{F,E} - k_{B,E}) \enp{\obs{\conc{\Go}}} +
    (k_{F,C} - k_{B,C}) \enp{\obs{\conc{\Gb}}} \biggr) \\[0.7em]
  & {}= \dfrac{(k_{F,C} + k_{B,C}) (k_{F,E} - k_{B,E})
  + (k_{F,E} + k_{B,E}) (k_{F,C} - k_{B,C})}{2 (k_{F,E} + k_{B,E} + k_{F,C} + k_{B,C})}
\end{align*}
This exact equation is derived in Ref.~\cite{Stukalin2005}.
We obtain it as a particular case of the general notion of
rate equations for graph explained below.
It is worth noting that,
despite the simplicity of the equation,
it is not easily derivable by hand.
This and other examples are available to play with in our web app
at \url{https://rhz.github.io/fragger/}.
The example models include
\begin{itemize}
\item the \href{https://rhz.github.io/fragger/?m=bimotor}{DNA walker model} described above;
\item a \href{https://rhz.github.io/fragger/?re=1&m=bunnies}{population model} tracking parent-child and sibling relationships;
\item the \href{https://rhz.github.io/fragger/?m=voter}{voter model}
  from Ref.~\cite{icfem};
\item the \href{https://rhz.github.io/fragger/?m=pa}{preferential attachment model} from Ref.~\cite{rc15}.
\end{itemize}
The DNA walker model presented in this introduction is small and reversible.
It requires no approximation to obtain a finite expansion.
By contrast, the population model and
the preferential attachment model are irreversible;
the population and the voter model require
an approximation to obtain a finite expansion.

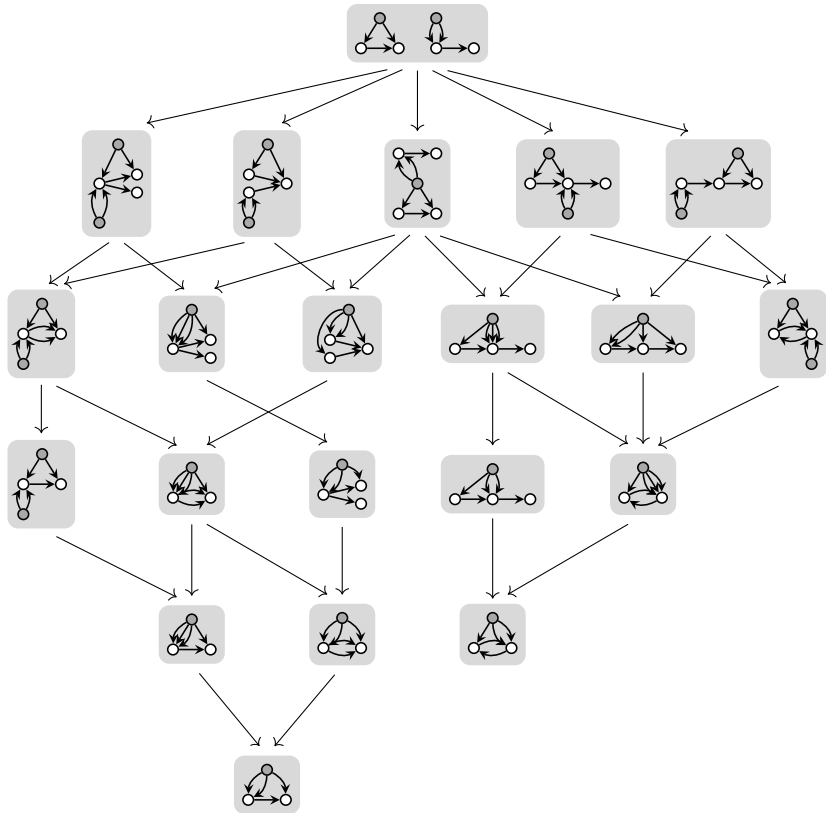
\begin{figure*}[tb]
  \vspace{-1cm}
  \begin{center}
    \begin{tikzpicture}[x=.5cm, y=.4cm]

      \node[grphnode] (mg-empty) at (0,0) {
        \tikz[ingrphdiag]{
          \n[b]{b1}{.5,1};
          \n[c]{c1}{0,0};
          \n[c]{c2}{1,0};
          \n[b]{b2}{2,1};
          \n[c]{c3}{2,0};
          \n[c]{c4}{3,0};
          \draw[->] (b1) -- (c1);
          \draw[->] (b1) -- (c2);
          \draw[->] (c1) -- (c2);
          \draw[->] (b2) edge[bend right] (c3);
          \draw[->] (b2) edge[bend left,LegColor] (c3);
          \draw[->] (c3) -- (c4);}};

\node[grphnode] (mg-b) at (0,-5) {
        \tikz[ingrphdiag]{
          \n[b]{b}{.5,1};
          \n[c]{c1}{0,0};
          \n[c]{c2}{1,0};
          \n[c]{c3}{0,2};
          \n[c]{c4}{1,2};
          \draw[->] (b) -- (c1);
          \draw[->] (b) -- (c2);
          \draw[->] (c1) -- (c2);
          \draw[->] (b) edge[bend right] (c3);
          \draw[->] (b) edge[bend left,LegColor] (c3);
          \draw[->] (c3) -- (c4);}};

      \arrsn{mg-empty}{mg-b};

\node[grphnode] (mg-c1) at (-8,-5) {
        \tikz[ingrphdiag]{
          \n[b]{b1}{.5,1.3};
          \n[c]{c1}{0,0};
          \n[c]{c2}{1,.3};
          \n[c]{c3}{1,-.3};
          \n[b]{b2}{0,-1.3};
          \draw[->] (b1) -- (c1);
          \draw[->] (b1) -- (c2);
          \draw[->] (c1) -- (c2);
          \draw[->] (b2) edge[bend right] (c1);
          \draw[->] (b2) edge[bend left,LegColor] (c1);
          \draw[->] (c1) -- (c3);}};

      \arrsn{mg-empty}{mg-c1};

\node[grphnode] (mg-c2) at (-4,-5) {
        \tikz[ingrphdiag]{
          \n[b]{b1}{.5,1.3};
          \n[c]{c1}{0,.3};
          \n[c]{c2}{1,0};
          \n[c]{c3}{0,-.3};
          \n[b]{b2}{0,-1.3};
          \draw[->] (b1) -- (c1);
          \draw[->] (b1) -- (c2);
          \draw[->] (c1) -- (c2);
          \draw[->] (b2) edge[bend right] (c3);
          \draw[->] (b2) edge[bend left,LegColor] (c3);
          \draw[->] (c3) -- (c2);}};

      \arrsn{mg-empty}{mg-c2};

\node[grphnode] (mg-c1*) at (8,-5) {
        \tikz[ingrphdiag]{
          \n[b]{b1}{1.5,1};
          \n[c]{c1}{0,0};
          \n[c]{c2}{1,0};
          \n[c]{c3}{2,0};
          \n[b]{b2}{0,-1};
          \draw[->] (b1) -- (c2);
          \draw[->] (b1) -- (c3);
          \draw[->] (c1) -- (c2);
          \draw[->] (b2) edge[bend right] (c1);
          \draw[->] (b2) edge[bend left,LegColor] (c1);
          \draw[->] (c2) -- (c3);}};

      \arrsn{mg-empty}{mg-c1*};

\node[grphnode] (mg-c2*) at (4,-5) {
        \tikz[ingrphdiag]{
          \n[b]{b1}{.5,1};
          \n[c]{c1}{0,0};
          \n[c]{c2}{1,0};
          \n[c]{c3}{2,0};
          \n[b]{b2}{1,-1};
          \draw[->] (b1) -- (c1);
          \draw[->] (b1) -- (c2);
          \draw[->] (c1) -- (c2);
          \draw[->] (b2) edge[bend right] (c2);
          \draw[->] (b2) edge[bend left,LegColor] (c2);
          \draw[->] (c2) -- (c3);}};

      \arrsn{mg-empty}{mg-c2*};

\node[grphnode] (mg-c1-b) at (-6,-10) {
        \tikz[ingrphdiag]{
          \n[b]{b}{.5,1.3};
          \n[c]{c1}{0,0};
          \n[c]{c2}{1,.3};
          \n[c]{c3}{1,-.3};
          \draw[->] (b) -- (c1);
          \draw[->] (b) -- (c2);
          \draw[->] (c1) -- (c2);
          \draw[->] (b) edge[bend right] (c1);
          \draw[->] (b) edge[bend left,LegColor] (c1);
          \draw[->] (c1) -- (c3);}};

      \arrsn{mg-c1}{mg-c1-b};
      \arrsn{mg-b}{mg-c1-b};

\node[grphnode] (mg-c2-b) at (-2,-10) {
        \tikz[ingrphdiag]{
          \n[b]{b}{.5,1.3};
          \n[c]{c1}{0,-.3};
          \n[c]{c2}{1,0};
          \n[c]{c3}{0,.3};
          \draw[->] (b) edge[bend right=70] (c1);
          \draw[->] (b) -- (c2);
          \draw[->] (c1) -- (c2);
          \draw[->] (b) edge[bend right] (c3);
          \draw[->] (b) edge[bend left,LegColor] (c3);
          \draw[->] (c3) -- (c2);}};

      \arrsn{mg-c2}{mg-c2-b};
      \arrsn{mg-b}{mg-c2-b};

\node[grphnode] (mg-c1*-b) at (6,-10) {
        \tikz[ingrphdiag]{
          \n[b]{b}{1,1};
          \n[c]{c1}{0,0};
          \n[c]{c2}{1,0};
          \n[c]{c3}{2,0};
          \draw[->] (b) -- (c2);
          \draw[->] (b) -- (c3);
          \draw[->] (c1) -- (c2);
          \draw[->] (b) edge[bend right=15] (c1);
          \draw[->] (b) edge[bend left=15,LegColor] (c1);
          \draw[->] (c2) -- (c3);}};

      \arrsn{mg-c1*}{mg-c1*-b};
      \arrsn{mg-b}{mg-c1*-b};

\node[grphnode] (mg-c2*-b) at (2,-10) {
        \tikz[ingrphdiag]{
          \n[b]{b}{1,1};
          \n[c]{c1}{0,0};
          \n[c]{c2}{1,0};
          \n[c]{c3}{2,0};
          \draw[->] (b) -- (c1);
          \draw[->] (b) -- (c2);
          \draw[->] (c1) -- (c2);
          \draw[->] (b) edge[bend right] (c2);
          \draw[->] (b) edge[bend left,LegColor] (c2);
          \draw[->] (c2) -- (c3);}};

      \arrsn{mg-c2*}{mg-c2*-b};
      \arrsn{mg-b}{mg-c2*-b};

\node[grphnode] (mg-c1-c2) at (-10,-10) {
        \tikz[ingrphdiag]{
          \n[b]{b1}{.5,1};
          \n[c]{c1}{0,0};
          \n[c]{c2}{1,0};
          \n[b]{b2}{0,-1};
          \draw[->] (b1) -- (c1);
          \draw[->] (b1) -- (c2);
          \draw[->] (b2) edge[bend right] (c1);
          \draw[->] (b2) edge[bend left,LegColor] (c1);
          \draw[->] (c1) edge[bend right] (c2);
          \draw[->] (c1) edge[bend left] (c2);}};

      \arrsn{mg-c1}{mg-c1-c2};
      \arrsn{mg-c2}{mg-c1-c2};

\node[grphnode] (mg-c1*-c2*) at (10,-10) {
        \tikz[ingrphdiag]{
          \n[b]{b1}{.5,1};
          \n[c]{c1}{0,0};
          \n[c]{c2}{1,0};
          \n[b]{b2}{1,-1};
          \draw[->] (b1) -- (c1);
          \draw[->] (b1) -- (c2);
          \draw[->] (b2) edge[bend right] (c2);
          \draw[->] (b2) edge[bend left,LegColor] (c2);
          \draw[->] (c1) edge[bend left] (c2);
          \draw[->] (c2) edge[bend left] (c1);}};

      \arrsn{mg-c1*}{mg-c1*-c2*};
      \arrsn{mg-c2*}{mg-c1*-c2*};

\node[grphnode] (mg-c1-c2-c1c2) at (-10,-15) {
        \tikz[ingrphdiag]{
          \n[b]{b1}{.5,1};
          \n[c]{c1}{0,0};
          \n[c]{c2}{1,0};
          \n[b]{b2}{0,-1};
          \draw[->] (b1) -- (c1);
          \draw[->] (b1) -- (c2);
          \draw[->] (b2) edge[bend right] (c1);
          \draw[->] (b2) edge[bend left,LegColor] (c1);
          \draw[->] (c1) -- (c2);}};

      \arrsn{mg-c1-c2}{mg-c1-c2-c1c2};

\node[grphnode] (mg-c1-b-bc1) at (-2,-15) {
        \tikz[ingrphdiag]{
          \n[b]{b}{.5,1};
          \n[c]{c1}{0,0};
          \n[c]{c2}{1,.3};
          \n[c]{c3}{1,-.3};
          \draw[->] (b) edge[bend right] (c1);
          \draw[->] (b) edge[bend left,LegColor] (c1);
          \draw[->] (b) edge[bend left] (c2);
          \draw[->] (c1) -- (c2);
          \draw[->] (c1) -- (c3);}};

      \arrsn{mg-c1-b}{mg-c1-b-bc1};

\node[grphnode] (mg-c2*-b-bc1) at (2,-15) {
        \tikz[ingrphdiag]{
          \n[b]{b}{1,1};
          \n[c]{c1}{0,0};
          \n[c]{c2}{1,0};
          \n[c]{c3}{2,0};
          \draw[->] (b) -- (c1);
          \draw[->] (c1) -- (c2);
          \draw[->] (b) edge[bend right] (c2);
          \draw[->] (b) edge[bend left,LegColor] (c2);
          \draw[->] (c2) -- (c3);}};

      \arrsn{mg-c2*-b}{mg-c2*-b-bc1};

\node[grphnode] (mg-c1-c2-b) at (-6,-15) {
        \tikz[ingrphdiag]{
          \n[b]{b}{.5,1};
          \n[c]{c1}{0,0};
          \n[c]{c2}{1,0};
          \draw[->] (b) -- (c1);
          \draw[->] (b) -- (c2);
          \draw[->] (b) edge[bend right] (c1);
          \draw[->] (b) edge[bend left,LegColor] (c1);
          \draw[->] (c1) edge[bend right] (c2);
          \draw[->] (c1) edge[bend left] (c2);}};

      \arrsn{mg-c2-b}{mg-c1-c2-b};
      \arrsn{mg-c1-c2}{mg-c1-c2-b};

\node[grphnode] (mg-c1*-c2*-b) at (6,-15) {
        \tikz[ingrphdiag]{
          \n[b]{b}{.5,1};
          \n[c]{c1}{0,0};
          \n[c]{c2}{1,0};
          \draw[->] (b) -- (c1);
          \draw[->] (b) -- (c2);
          \draw[->] (b) edge[bend right] (c2);
          \draw[->] (b) edge[bend left,LegColor] (c2);
          \draw[->] (c1) edge[bend left] (c2);
          \draw[->] (c2) edge[bend left] (c1);}};

      \arrsn{mg-c1*-c2*}{mg-c1*-c2*-b};
      \arrsn{mg-c2*-b}{mg-c1*-c2*-b};
      \arrsn{mg-c1*-b}{mg-c1*-c2*-b};

\node[grphnode] (mg-c1-c2-b-c1c2) at (-6,-20) {
        \tikz[ingrphdiag]{
          \n[b]{b}{.5,1};
          \n[c]{c1}{0,0};
          \n[c]{c2}{1,0};
          \draw[->] (b) -- (c1);
          \draw[->] (b) -- (c2);
          \draw[->] (b) edge[bend right] (c1);
          \draw[->] (b) edge[bend left,LegColor] (c1);
          \draw[->] (c1) -- (c2);}};

      \arrsn{mg-c1-c2-b}{mg-c1-c2-b-c1c2};
      \arrsn{mg-c1-c2-c1c2}{mg-c1-c2-b-c1c2};

\node[grphnode] (mg-c1-c2-b-bc1) at (-2,-20) {
        \tikz[ingrphdiag]{
          \n[b]{b}{.5,1};
          \n[c]{c1}{0,0};
          \n[c]{c2}{1,0};
          \draw[->] (b) edge[bend right] (c1);
          \draw[->] (b) edge[bend left,LegColor] (c1);
          \draw[->] (b) edge[bend left] (c2);
          \draw[->] (c1) edge[bend right] (c2);
          \draw[->] (c1) edge[bend left] (c2);}};

      \arrsn{mg-c1-c2-b}{mg-c1-c2-b-bc1};
      \arrsn{mg-c1-b-bc1}{mg-c1-c2-b-bc1};

\node[grphnode] (mg-c1*-c2*-b-bc1) at (2,-20) {
        \tikz[ingrphdiag]{
          \n[b]{b}{.5,1};
          \n[c]{c1}{0,0};
          \n[c]{c2}{1,0};
          \draw[->] (b) -- (c1);
          \draw[->] (b) edge[bend right] (c2);
          \draw[->] (b) edge[bend left,LegColor] (c2);
          \draw[->] (c1) edge[bend left] (c2);
          \draw[->] (c2) edge[bend left] (c1);}};

      \arrsn{mg-c1*-c2*-b}{mg-c1*-c2*-b-bc1};
      \arrsn[.05]{mg-c2*-b-bc1}{mg-c1*-c2*-b-bc1};

\node[grphnode] (mg-c1-c2-b-bc1-c1c2) at (-4,-25) {
        \tikz[ingrphdiag]{
          \n[b]{b}{.5,1};
          \n[c]{c1}{0,0};
          \n[c]{c2}{1,0};
          \draw[->] (b) edge[bend right] (c1);
          \draw[->] (b) edge[bend left,LegColor] (c1);
          \draw[->] (b) edge[bend left] (c2);
          \draw[->] (c1) -- (c2);}};

      \arrsn{mg-c1-c2-b-bc1}{mg-c1-c2-b-bc1-c1c2};
      \arrsn{mg-c1-c2-b-c1c2}{mg-c1-c2-b-bc1-c1c2};

    \end{tikzpicture}
  \end{center}
  \caption{The poset of minimal gluings of $G_2$ and $G_1$.
The disjoint sum is at the top.
    Gluings are layered by the number
    of node and edge identifications or,
    equivalently, by the size of their intersection.}\label{fig:bimotor-mgs}
\end{figure*}

\subsection{Discussion}
\label{sec:discussion}
The reasoning and derivation done above
in the DNA walker can be made completely general. Given a \emph{graph observable} $\obs{F}$, meaning a
function counting the number of embeddings of the graph $F$ in the state $X$, one can build an ODE which describes the rate at
which the mean occurrence count $\exp(\obs{F}\hskip-0.3ex(X(t)))$ changes over
time.

Because the underlying Markov process $X(t)$ is generated by graph-rewriting rules,
the one combinatorial ingredient to build that equation is the notion of \emph{minimal gluings} (MGs) of a pair of graphs. Terms in the ODE for $F$ are derived from the set of MGs of $F$ with the left and right sides of the rules which generate $X(t)$.
Besides, each term in $F$'s ODE depends on the current state \emph{only} via expressions of the form $\exp(\obs{G})$ for $G$ a graph defining a new observable.
Thus each fresh observable $\obs{G}$ can then be submitted to the same treatment, and one obtains in general a countable system of rate equations for graphs.
In good cases (as in the walker example), the expansion is finite
and there is no need for any approximation.
In general, one needs to truncate the expansion.
As the MFA expansion is a symbolic procedure one can
pursue it in principle to any order.

The significance of the method hinges both on how many models can be captured in graph-like languages, and how accurate the obtained MFAs are. While these models do no exhaust all possibilities, GTSs seem very expressive.
In addition,
our approach to the derivation of rate equations for graphs uses a
general categorical treatment which subsumes various graph-like
structures such as: hyper-graphs, typed graphs, etc.~\cite{BaldanCHKS14mscs,LackS05ita}.
This abstract view is mathematically convenient, and broadens the set of models to which the method applies.

What we know about the existence of solutions to the (in general) countable ODE systems generated by our method is limited. For general countable
continuous-time Markov chains and observables, existence of a solution is not guaranteed~\cite{spieksma}. Despite their great popularity, the current mathematical understanding of the quality of MFAs only addresses the case of CRNs and density-dependent Markov chains, with Kurtz' theory of
scalings~\cite[Chap.~11]{kurtz}, or the case of dynamics on static graphs~\cite{gleeson}. Some progress on going beyond the formal point of view and obtaining existence theorems for solutions of REs for graphs were reported in Ref.~\cite{Tobiasfossacs}.
Another limitation is the accuracy of MFAs once truncated (as they must be if one wants to plug them in an ODE solver).
Even if an MFA can be built to any desired order, it might still fall short of giving a sensible picture of the dynamics of interest. Finally, it may also be that the cost of running a convincing approximation is about the same as that of simulating the system upfront.

\subsection{Related work}
This paper follows ideas on applying the methods of abstract
interpretation to the differential semantics of \emph{site graph}
rewriting~\cite{chaosruss,sasbsandro,pnasjerome}.
Another more remote influence is Lynch's finite-model theoretic
approach to MFAs~\cite{lynch}.
From the GTS side, the theory of site graph rewriting had long been
thought to be a lucky anomaly until a recent series of work showed
that most of its ingredients could be made sense of, and given a much
larger basis of applications, through the use of algebraic
graph-rewriting
techniques~\cite{jonathantobias,BapodraH10eceasst,Heckel12lncs}.
These latter investigations motivated us to try to address MFA-related questions at a higher level of generality.

\paragraph{Relation to the rule-algebraic approach.}
Another approach to the same broad set of questions started with 
Behr et al.\@ introducing ideas from representation theory 
commonly used in statistical physics~\cite{BehrDG16lics}. They show that one can
study the algebra of rules and their composition law in a way that is decoupled from what is actually 
being re-written. The ``rule-algebraic'' theory allows one to derive Kolmogorov equations for observables based on a systematic use of rule commutators (recently implemented in Ref.~\cite{behr2020commutators}). Interestingly, novel notions of graph rewriting appear~\cite{behr2016algebras}. Partial differential equations describing the generating function of such observables can be derived systematically~\cite{lmcs:6628}. As the theory can handle adhesive categories in general and sesqui-pushout rewriting~\cite{behr2019sesqui}, it offers an treatment of irreversible rewrites alternative to the one presented in this paper. (The rule-algebraic  approach can also handle application conditions~\cite{behr2019compositionality}). It will need further work to precisely pinpoint how these two threads of work articulate both at the theoretical and at the implementation levels.
 
\paragraph{Outline.}
The paper is organised as follows:
\Sec{rewrite} collects preliminaries on graph-rewriting and
establishes the key \emph{forward} and \emph{backward modularity}
lemmas;
\Sec{diff-sem}
derives our main result namely a concrete formula for the action of a
generator associated to a set  of graph-rewriting rules as specified in \Sec{rewrite}. From this formula, the rate equation for graphs follows easily.
Basic category-theoretical definitions needed in the main text are given in \App{defs}; axiomatic proofs in \App{proofs}.

\section{Stochastic graph rewriting}
\label{sec:rewrite}

We turn now to the graphical framework within which we will carry out
the derivation of our generalised rate equation (GREG)
in~\Sec{diff-sem}.
We use a categorical approach know as algebraic graph rewriting,
specifically the \emph{single pushout (SPO)}
approach~\cite{Lowe93tcs,EhrigHKLRWC97gg}.
The reasons for this choice are twofold:
first, we benefit from a solid body of preexisting work;
second, it allows for a succinct and `axiomatic' presentation
abstracting over the details of the graph-like structures that are
being rewritten.
Working at this high level of abstraction allows us to identify a set
of generic properties necessary for the derivation of the GREG without
getting bogged down in the details of the objects being rewritten.
Indeed, while we only treat the case of directed multigraphs (graphs
with an arbitrary number of directed edges between any two nodes) in
this section, the proofs of all lemmas are set in the more general
context of \emph{adhesive categories}~\cite{LackS05ita}
in~\App{proofs}.
This extends the applicability of our technique to rewrite systems
over typed graphs and hypergraphs, among others.

For the convenience of the reader, we reproduce from Ref.~\cite{rc15}
our basic definitions for the category $\Grph$ of \emph{directed
  multigraphs}.
Next, we briefly summarise the SPO approach and its \emph{stochastic
  semantics}~\cite{HeckelLM06fuin}.
We conclude with the \emph{modularity lemmas}, which are key to
the derivation of the GREG in the next section.

\subsection{The category of directed multigraphs}

A \emph{directed multigraph} $G$ consists of
a finite set of \emph{nodes} $V_G$, a finite set of \emph{edges}
$E_G$, and \emph{source} and \emph{target} maps $\src[G], \tgt[G]
\colon E_G \to V_G$.
A \emph{graph morphism} $f \colon G \to H$ between graphs $G$ and $H$
is a pair of maps $f_E \colon E_G \to E_H$, $f_V \colon V_G \to V_H$
which preserve the graph structure, \ie such that for all $e \in E_G$,
\begin{align*}
  \src[H](f_E(e)) &= f_V(\src[G](e)) &\text{ and }& &
  \tgt[H](f_E(e)) &= f_V(\tgt[G](e)).
\end{align*}
The graphs $G$ and $H$ are called the \emph{domain} and
\emph{codomain} of~$f$.
A graph morphism~$f \colon G \to H$ is a \emph{monomorphism}, or
simply a \emph{mono}, if $f_V$ and $f_E$ are injective;
it is a \emph{graph inclusion} if both $f_V$ and $f_E$ are inclusion
maps, in which case $G$ is a \emph{subgraph} of $H$ and we write
$G \subseteq H$.
Every morphism $f \colon G \to H$ induces a subgraph
$f(G) \subseteq H$ called the \emph{direct image} (or just the
\emph{image}) of $f$ in~$H$, such that $V_{f(G)} = f_V(V_G)$ and
$E_{f(G)} = f_E(E_G)$.
\Fig{ex} illustrates a graph and a graph morphism.

\begin{figure}[tb]
  \centering
  \begin{subfigure}[c]{.35\textwidth}\centering
    \begin{tikzpicture}[grphdiag, x=.8cm, y=.6cm]
      \n[c]{c2}{.5,1};
      \n[c]{c3}{0,0};
      \n[c]{c4}{1,0};
      \n[c]{c5}{2,0};
      \draw[->] (c2) -- (c3);
      \draw[->] (c2) -- (c4);
      \draw[->] (c3) -- (c4);
      \draw[->] (c4) -- (c5);
      \draw (c4) edge[->, bend left] (c5);
    \end{tikzpicture}
    \phantomsubcaption\label{fig:ex-grph}
  \end{subfigure}
  \begin{subfigure}[c]{.5\textwidth}\centering
    \begin{tikzpicture}[x=1cm, y=.8cm, thick]
\begin{scope}[>=stealth]
        \n[b]{b2}{.5,1};
        \n[c]{c2}{0,0};
        \n[c]{c3}{1,0};
        \n[c]{c4}{2,0};
        \draw[->, blue]     (b2) -- coordinate(d2)          (c2);
        \draw[->, BrickRed] (b2) -- coordinate[pos=0.4](r3) (c3);
        \draw[->]           (c2) -- coordinate(e2)          (c3);
        \draw[->]           (c3) -- coordinate(e3)          (c4);
      \end{scope}

\begin{scope}[>=stealth, xshift=4cm]
        \n[b]{b1}{0,1};
        \n[c]{c1}{0,0};
        \draw (b1) edge[->, bend right, BrickRed] coordinate(r1) (c1);
        \draw (b1) edge[->, bend left, blue]      coordinate(d1) (c1);
        \draw (c1) edge[->, loop below]           coordinate[pos=.8](e1)
          coordinate[pos=.5](e1b)(c1);
      \end{scope}

\begin{scope}[every edge/.style={
          semithick, draw=black, densely dotted, -to}]
        \draw (b2) edge[bend left] (b1);
        \draw (c2) edge[bend left] (c1);
        \draw (c3) edge[bend left] (c1);
        \draw (c4) edge            (c1);
      \end{scope}

\begin{scope}[every edge/.style={
          semithick, draw=DarkGreen, densely dotted, -to}]
        \draw (e2) edge[bend right] (e1);
        \draw (e3) edge[bend right] (e1);
        \draw (d2) edge[bend left]  (d1);
        \draw (r3) edge[bend left]  (r1);
      \end{scope}

\begin{scope}[on background layer]
        \node[grphdiag-bg, inner sep=4pt, fit=(b1)(e1b)]{};
        \node[grphdiag-bg, inner sep=4pt, fit=(b2)(c2)(c4)]{};
      \end{scope}
    \end{tikzpicture}
    \phantomsubcaption\label{fig:ex-hom}
    \vspace{-\baselineskip}
  \end{subfigure}
  \caption{Examples of \subref{fig:ex-grph}) a directed multigraph,
    \subref{fig:ex-hom}) a graph morphism.}
  \label{fig:ex}
\end{figure}
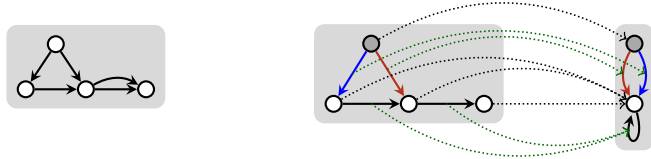

A \emph{partial} graph morphism $p \colon G \rto H$ is a pair of
partial maps $p_V \colon V_G \rto V_H$ and
$p_E \colon E_G \rto E_H$ that preserve the graph structure.
Equivalently,
$p$ can be represented as a \emph{span} of (total) graph morphisms,
that is, a pair of morphisms $p_1 \colon K \to G$,
$p_2 \colon K \to H$ with common domain $K$, where $p_1$ is mono
and $K$ is the \emph{domain of definition} of $p$.  We will
use whichever representation is more appropriate for the task at hand.
Graphs and graph morphisms form the category $\Grph$
(with the obvious notion of composition and identity) while graphs and
partial graph morphisms form the category $\GrphP$.

Graph morphisms provide us with a notion of {pattern matching} on
graphs while partial graph morphisms provide the accompanying notion of
{rewrite rule}.
We restrict pattern matching to monos:
a \emph{match} of a pattern $L$ in a graph $G$
is a monomorphism $f \colon L \to G$.
We write $\mchs{L}{G}$ for the set of matches of $L$ in~$G$.
We also restrict rules:
a \emph{rule} is a partial graph morphism
$\al \colon L \rto R = \tuple{\al_1 \colon K \to L, \al_2 \colon K \to
  R}$ where both $\al_1$ and $\al_2$ are monos.
We say that $L$ and $R$ are $\al$'s left and right hand side (LHS and
RHS).
Rules are special cases of partial graph morphisms and compose as
such.
Given a rule $\al \colon L \rto R = \tuple{\al_1, \al_2}$, we define
the \emph{reverse} rule $\dg{\al} \colon R \rto L$ as the pair
$\dg{\al} := \tuple{\al_2, \al_1}$,
not to be confused with the inverse of $\al$ (which does not exist in
general).
Note that $\dg{-}$ is an involution, that is, $\dg{(\dg{\al})} = \al$.

\subsection{Graph rewriting}

The basic rewrite steps of a GTS are called \emph{derivations}.
We first describe them informally.  \Fig{deriv} shows a commutative square,
with a match $f \colon L \to G$ on the left and a rule $\al \colon L
\rto R$, on top.  The match $f$ identifies the subgraph in $G$ that is to be modified,
while the rule $\al$ describes how to carry out the modification.  In
order to obtain the \emph{comatch} $g \colon R \to H$ on the right,
one starts by
removing nodes and edges from $f(L)$ which do not have a preimage
under $f \comp \al_1$, as well as any edges left dangling (coloured
red in the figure).  To complete the derivation, one extends the
resulting match by adjoining to $D$ the nodes and edges in $R$ that do
not have a preimage under $\al_2$ (coloured green in the figure).

\newcommand{\minigrphdiag}[2]{\node (#1) {\tikz[grphdiag, x=.8cm, y=.8cm, scale=.6,n/.append style={minimum size=6pt}]{#2}$_{#1}$};}
\begin{figure}[tb]
  \centering
  \begin{tikzpicture}[>=cm to]
\matrix[inner sep=3pt, column sep=2.5em, row sep=2em]{
      \minigrphdiag{L}{
        \n[r]{c2}{0,0};
        \n[c]{c3}{1,0};
        \n[p]{c4}{2,0};
        \draw[->, BrickRed] (c2) -- (c3);
      } &
      \minigrphdiag{K}{
        \n[p]{c2}{0,0};
        \n[c]{c3}{1,0};
        \n[p]{c4}{2,0};
      } &
      \minigrphdiag{R}{
        \n[p]{c2}{0,0};
        \n[c]{c3}{1,0};
        \n[g]{c4}{2,0};
        \draw[->, DarkGreen] (c3) -- (c4);
      }\\
      \minigrphdiag{G}{
        \n[c]{c1}{.5,1};
        \n[r]{c2}{0,0};
        \n[c]{c3}{1,0};
        \n[p]{c4}{2,0};
        \draw[->, BrickRed] (c1) -- (c2);
        \draw[->] (c1) -- (c3);
        \draw[->, BrickRed] (c2) -- (c3);
      } &
      \minigrphdiag{D}{
        \n[c]{c1}{.5,1};
        \n[p]{c2}{0,0};
        \n[c]{c3}{1,0};
        \n[p]{c4}{2,0};
        \draw[->] (c1) -- (c3);
      } &
      \minigrphdiag{H}{
        \n[c]{c1}{.5,1};
        \n[p]{c2}{0,0};
        \n[c]{c3}{1,0};
        \n[g]{c4}{2,0};
        \draw[->] (c1) -- (c3);
        \draw[->, DarkGreen] (c3) -- (c4);
      }\\
    };

\draw[->] (L) -- node[left]{$f$} (G);
    \draw[->] (K) -- node[above]{$\al_1$} (L);
    \draw[->] (K) -- (D);
    \draw[->] (K) -- node[above]{$\al_2$} (R);
    \draw[->] (R) -- node[right]{$g$} (H);
    \draw[->] (D) -- node[above]{$\beta_1$} (G);
    \draw[->] (D) -- node[above]{$\beta_2$} (H);
  \end{tikzpicture}
  \caption{A derivation.}
  \label{fig:deriv}
\end{figure}
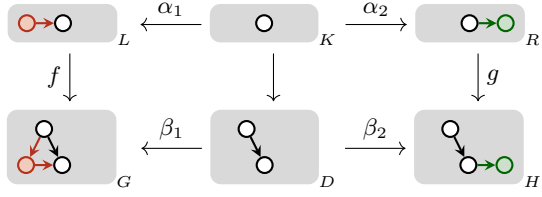

Derivations constructed in this way have the defining property of
\emph{pushout squares}  (PO) in~$\GrphP$, hence the name SPO
for the approach.
Alternatively, one can describe a derivation through the
properties of its inner squares: the left square is the
\emph{final pullback complement} (FPBC) of $\al_1$ and $f$,
while the right one is a PO in $\Grph$~\cite{CorradiniHHK06icgt}.
(Definitions and basic properties of POs and FPBCs are given in \App{defs}.)
\begin{definition}\label{def:deriv}
  A \emph{derivation} of a \emph{comatch} $g \colon R \to H$ from a
  match $f \colon L \to G$ by a rule
$\al = \tuple{\al_1\colon K \to L, \al_2\colon K \to R}$ is a
  diagram in $\Grph$ such as~\eqref{dg:deriv1}, where the left square
  is an FPBC of $f$ and $\al_1$ and the right square is a PO,
  \begin{diags}
    \begin{tikzcd}
      L\dar[swap]{f} & K \lar[swap]{\al_1}\rar{\al_2}\dar{h} &
      R \dar{g}\\
      G & D \lar[swap]{\beta_1} \rar{\beta_2} & H
    \end{tikzcd}~\eqnlabel{dg:deriv1}\hspace{3em}
    \begin{tikzcd}
      L\dar[swap]{f}\rar[rule]{\al} & R \dar{g}\\
      G \rar[rule]{\beta} & H
    \end{tikzcd}~\eqnlabel{dg:deriv2}
  \end{diags}
  with $h$, $g$ matches and $\ba = \tuple{\ba_1, \ba_2}$ a rule,
  called the \emph{corule} of the derivation.
\end{definition}
Equivalently, a derivation of $g$ from $f$ by $\al$ is a PO
in~$\GrphP$ as in~\eqref{dg:deriv2}, with corule~$\ba$.
We will mostly use this second characterisation of derivations.

Write $f \rwto[\al] g$ if there is a derivation of $g$ from $f$ by
$\al$.
Since derivations are POs of partial morphisms and $\GrphP$ has all
such POs~\cite{Lowe93tcs}, the relation $\rwto[\al]$~is \emph{total},
that is, for any match $f$ and rule $\al$ (with common domain), we
can find a comatch $g$.
However, the converse is not true:
not every match $g$ having the RHS of $\alpha$ as its domain is a
comatch of a derivation by $\al$.
Which is to say, there might not exist $f$ such that $f \rwto[\al] g$
(the relation $\rwto[\al]$ is not surjective).
When there is such an $f$, we say $g$ is \emph{derivable} by $\al$.
Consider the example in \Fig{deriv-ex}.
Here, $g$ is $\al$-derivable (as witnessed by $f$) but $h$ is not:
no match of the LHS could contain a ``preimage'' of the extra (red)
edge $e$ in the codomain of $h$ because the target node of
$e$ has not yet been created.

We say a derivation $f \rwto[\al] g$ (with corule $\ba$) is
\emph{reversible} if $g \rwto[\dg{\al}] f$ (with corule $\dg{\ba}$),
and \emph{irreversible} otherwise.
Clearly, derivations are not reversible in general, otherwise
$\rwto[\al]$ would be surjective.
Consider the derivation shown in \Fig{deriv}.
The derivation removes two (red) edges from the codomain of~$f$;
the removal of the lower edge is specified in the LHS of $\al$,
whereas the removal of the upper edge is a \emph{side effect} of
removing the red node to which the edge is connected (graphs cannot
contain dangling edges).
Applying the reverse rule $\dg{\al}$ to the comatch $g$ restores the
red node and the lower red edge, but not the upper red edge.
In other words, $f$ is not $\dg{\al}$-derivable, hence the derivation
in \Fig{deriv} is irreversible.
In previous work, we have shown how to derive rate equations for graph
transformation systems with only reversible
derivations~\cite{sasbsandro,icfem,rc15}.
In \Sec{diff-sem}, we overcome this limitation, giving a procedure
that extends to the irreversible case.

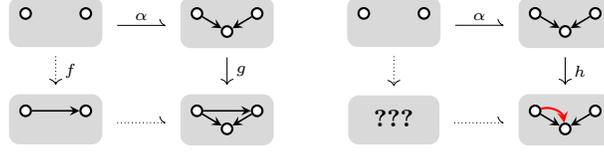
\begin{figure}[tb]
  \centering
  \begin{tikzcd}
    \conc{\n[c]{p1}{0, 0.9};
      \n[c]{p2}{2, 0.9};
      \n[p]{c}{1, 0.3};
    } \rar[rule]{\al}
    \dar[densely dotted]{\, f} &[1ex]
    \conc{\n[c]{p1}{0, 0.9};
      \n[c]{p2}{2, 0.9};
      \n[c]{c}{1, 0.3};
      \draw[->] (p1) -- (c);
      \draw[->] (p2) -- (c);
    } \dar{\, g}\\[-3pt]
    \conc{\n[c]{p1}{0, 0.9};
      \n[c]{p2}{2, 0.9};
      \n[p]{c}{1, 0.3};
      \draw[->] (p1) -- (p2);
    } \rar[rule, densely dotted] & \conc{\n[c]{p1}{0, 0.9};
      \n[c]{p2}{2, 0.9};
      \n[c]{c}{1, 0.3};
      \draw[->] (p1) -- (p2);
      \draw[->] (p1) -- (c);
      \draw[->] (p2) -- (c);
    }
  \end{tikzcd}\hspace{1.5em}
  \begin{tikzcd}
    \conc{\n[c]{p1}{0, 0.9};
      \n[c]{p2}{2, 0.9};
      \n[p]{c}{1, 0.3};
    } \rar[rule]{\al} \dar[densely dotted] &[1ex]
    \conc{\n[c]{p1}{0, 0.9};
      \n[c]{p2}{2, 0.9};
      \n[c]{c}{1, 0.3};
      \draw[->] (p1) -- (c);
      \draw[->] (p2) -- (c);
    } \dar{\, h}\\[-3pt]
    \conc{\n[p]{p1}{0, 0.9};
      \n[p]{p2}{2, 0.9};
      \n[p]{c}{1, 0.3};
      \node[inner sep=0] at (1, 0.4) {\bf ???};
    }  \rar[rule, densely dotted] & \conc{\n[c]{p1}{0, 0.9};
      \n[c]{p2}{2, 0.9};
      \n[c]{c}{1, 0.3};
      \draw (p1) edge[->, bend left=50, thick, red] (c);
      \draw[->] (p1) -- (c);
      \draw[->] (p2) -- (c);
    }
  \end{tikzcd}
  \caption{The match $g$ is $\al$-derivable, while $h$ is not.}
  \label{fig:deriv-ex}
\end{figure}

Since POs are unique (up to unique isomorphism), $\rwto[\al]$ is also
\emph{functional} (up to isomorphism).
The fact that derivations are only defined up to isomorphism is
convenient as it allows us to manipulate them without paying attention
to the concrete naming of nodes and edges.
Without this flexibility, stating and proving properties such as
\Lem[Lemmas]{bmod} and \ref{lem:derivability} below would be
exceedingly cumbersome.
On the other hand, when defining the stochastic semantics of our
rewrite systems, it is more convenient to restrict $\rwto[\al]$ to a
properly functional relation.
To this end, we fix once and for all, for any given match
$f \colon L \to G$ and rule $\al \colon L \rto R$, a
\emph{representative} $f \rwto[\al] \al(f)$ from the corresponding
isomorphism class of derivations, with (unique) comatch
$\al(f) \colon R \to H$, and (unique) corule $f(\al) \colon G \rto H$.

A set of rules $\rset$ thus defines a \emph{labelled transition system
  (LTS)} over graphs, with corules as transitions, labelled by the
associated pair $\tuple{f, \al}$.
Given a rule $\al \colon L \rto R$, we define a stochastic rate matrix
$Q_\al := (q^\al_{G H})$ over graphs as follows.
\begin{alignat}{2}
  q^\al_{G H} & :=
  \abs{\setof{f \in \mchs{L}{G}}{\al(f) \in \mchs{R}{H}}}
  &\qquad& \text{ for $G \neq H$,}
  \nonumber\\
  q^\al_{G G} & := \tsum_{H \neq G} -q^\al_{G H}
  && \text{ otherwise.}
  \label{eq:rates}
\end{alignat}
Given a \emph{model}, that is to say a finite set of rules $\rset$ and
a \emph{rate map} $k\colon \rset \to \RR^+$, we define the model rate
matrix $Q(\rset,k)$ as
\begin{align}
  Q(\rset,k) & := \tsum_{\al\in \rset} k(\al) Q_\al
\end{align}
Thus a model defines a CTMC over $\GrphP$.
As $\rset$ is finite, $Q(\rset,k)$ is row-finite.

\subsection{Composition and modularity of derivations}
\label{sec:deriv-prop}

By the well-known Pushout Lemma,
derivations can be composed horizontally (rule composition) and
vertically (rule specialisation) in the sense that if inner
squares below are derivations, so are the outer ones:
\begin{diags}\begin{tikzcd}
    L \rar[rule]{\al_1}\dar[swap]{f} &
    R_1 \rar[rule]{\al_2}\dar{g_1} & R_2 \dar{g_2} \\
    G \rar[rule] & H_1 \rar[rule] & H_2
  \end{tikzcd}\hspace{3em}
  \begin{tikzcd}[row sep=1.4em]
    L \rar[rule]{\al_1}\dar[swap]{f_1} & R \dar{g_1} \\
    G_1 \rar[rule]{\al_2}\dar[swap]{f_2} & H_1 \dar{g_2} \\
    G_2 \rar[rule] & H_2
  \end{tikzcd}\end{diags}
Derivations can also be decomposed vertically.  First, one has
a forward decomposition (which follows immediately from pasting of POs
in~$\GrphP$):
\begin{lemma}[Forward modularity]\label{lem:fmod}
  Let $\al$, $\ba$, $\gamma$ be rules and
  $f_1$, $f_2$, $g$, $g_1$ matches such that diagrams~\eqref{dg:fmod1} and \eqref{dg:fmod2} are
  derivations.
Then there is a unique match $g_2$ such that
  diagram~\eqref{dg:fmod3} commutes (in~$\GrphP$)
  and is a vertical composition of derivations.
\begin{diags}
    \begin{tikzcd}
      L \dar[swap]{f_1}\rar[rule]{\al} & R \arrow{dd}{g}\\
      S \dar[swap]{f_2}\\
      G \rar[rule]{\beta}              & H
    \end{tikzcd}~\eqnlabel{dg:fmod1}\hspace{1em}
    \begin{tikzcd}
      L \dar[swap]{f_1}\rar[rule]{\al}    & R \dar{g_1}\\
      S \rar[rule]{\gamma}                & T
    \end{tikzcd}~\eqnlabel{dg:fmod2}\hspace{1em}
    \begin{tikzcd}
      L \dar[swap]{f_1}\rar[rule]{\al}    & R \dar[swap]{g_1}\arrow[bend left=40]{dd}{g}\\
      S \dar[swap]{f_2}\rar[rule]{\gamma} & T \dar[swap, dashed]{g_2}\\
      G \rar[rule]{\beta}                 & H
    \end{tikzcd}~\eqnlabel{dg:fmod3}
  \end{diags}
\end{lemma}
A novel observation, which will play a central role in the next
section, is that one also has a backward decomposition:
\begin{lemma}[Backward modularity]\label{lem:bmod}
Let $\al$, $\ba$, $\gamma$ be rules and
  $f$, $f_1$, $g_1$, $g_2$ matches such that diagrams~\eqref{dg:bmod1} and \eqref{dg:bmod2} are
  derivations.
Then there is a unique match $f_2$ such that diagram~\eqref{dg:bmod3} commutes (in~$\GrphP$)
  and is a vertical composition of derivations.
  \begin{diags}
    \begin{tikzcd}
      L \arrow{dd}[swap]{f}\rar[rule]{\al} & R \dar{g_1}\\
                                           & T \dar{g_2}\\
      G \rar[rule]{\beta}                  & H
    \end{tikzcd}~\eqnlabel{dg:bmod1}\hspace{1em}
    \begin{tikzcd}
      L \dar[swap]{f_1} & R \dar{g_1}\lar[elur, swap]{\; \dg{\al}}\\
      S                 & T \lar[elur, swap]{\; \dg{\gamma}}
    \end{tikzcd}~\eqnlabel{dg:bmod2}\hspace{1em}
    \begin{tikzcd}
      L \arrow[bend right=40]{dd}[swap]{f}\rar[rule]{\al}\dar{f_1} &
      R \dar{g_1}\\
      S \dar[dashed]{f_2}\rar[rule]{\gamma} & T \dar{g_2}\\
      G \rar[rule]{\beta}                   & H
    \end{tikzcd}~\eqnlabel{dg:bmod3}
  \end{diags}
\end{lemma}
Forward and backward modularity look deceptively similar, but while
\Lem{fmod} is a standard property of POs,
\Lem{bmod} is decidedly non-standard.
Remember that derivations are generally irreversible.
It is therefore not at all obvious that one should be able to
transport factorisations of comatches backwards along a rule, let
alone in a unique fashion.
Nor is it obvious that the top half of the resulting decomposition
should be reversible.
The crucial ingredient that makes backward modularity possible is that
both matches and rules are monos.
Because rules are (partial) monos, we can reverse $\al$ and $\ba$
in~\eqref{dg:bmod1}, and the resulting diagram still commutes (though
it is no longer a derivation in general).
The existence and uniqueness of $f_2$ is then a direct consequence of
the universal property of~\eqref{dg:bmod2}, seen as a PO.
The fact that~\eqref{dg:bmod2} is reversible relies on matches also
being monos, but in a more subtle way.
Intuitively, the graph $T$ cannot contain any superfluous edges of the
sort that render the derivation in~\Fig{deriv} irreversible because,
$g_2$ being a mono, such edges would appear in $H$ as subgraphs,
contradicting the $\al$-derivability of $g_2 \comp g_1$.
Together, the factorisation of $f$ and the reversibility
of~\eqref{dg:bmod2} then induce the decomposition in~\eqref{dg:bmod3}
by \Lem{fmod}.
A full, axiomatic proof of \Lem{bmod} is given in \App{bmod}.

Among other things, \Lem{bmod} allows one to relate
\emph{derivability} of matches to \emph{reversibility} of derivations:
\begin{lemma}\label{lem:derivability}
  A match $g \colon R \to H$ is derivable by a rule $\al \colon L
  \rto R$ if and only if the derivation $g \rwto[\dg{\al}] f$ is
  reversible.\end{lemma}

\subsection{Gluings}

Given $G_1 \subseteq H$ and $G_2 \subseteq H$, the \emph{union} of
$G_1$ and $G_2$ in $H$ is the unique subgraph $G_1 \cup G_2$
of $H$, such that $V_{(G_1 \cup G_2)} = V_{G_1} \cup V_{G_2}$
and $E_{(G_1 \cup G_2)} = E_{G_1} \cup E_{G_2}$.  The
\emph{intersection} $(G_1 \cap G_2) \subseteq H$ is defined
analogously.  The subgraphs of $H$ form a complete distributive
lattice with $\cup$ and $\cap$ as the join and meet operations.
One can \emph{glue} arbitrary graphs as follows:
\begin{definition}\label{def:mgs}
  A \emph{gluing} of graphs $G_1$, $G_2$ is
a pair of matches
  $i_1 \colon G_1 \to U$, $i_2 \colon G_2 \to U$ with common
  codomain $U$; if in addition $U = i_1(G_1) \cup i_2(G_2)$,
  one says the gluing is \emph{minimal}.
\end{definition}
Two gluings
$i_1 \colon G_1 \to U$, $i_2 \colon G_2 \to U$ and
$j_1 \colon G_1 \to V$, $j_2 \colon G_2 \to V$
are said to be \emph{isomorphic}
if there is an isomorphism
$u \colon U \to V$, such that
$j_1 = u \comp i_1$ and
$j_2 = u \comp i_2$.
We write
$\MGCls{G_1}{G_2}$ for the set of isomorphism classes of minimal
gluings (MG) of $G_1$ and $G_2$, and
$\MGs{G_1}{G_2}$ for an arbitrary choice of representatives from each class in
$\MGCls{G_1}{G_2}$.
Given a gluing $\mu \colon G_1 \to H \gets G_2$, denote by $\hat{\mu}$
its ``tip'', \ie the common codomain $\hat{\mu} = H$ of $\mu$.

It is easy to see the following (see \App{proofs}
for an axiomatic proof):
\begin{lemma}\label{lem:mgs}
  Let $G_1$, $G_2$ be graphs, then
$\MGs{G_1}{G_2}$ is finite, and
for every gluing
  $f_1 \colon G_1 \to H$, $f_2 \colon G_2 \to H$,
there is a unique MG
$i_1 \colon G_1 \to U$, $i_2 \colon G_2 \to U$
  in $\MGs{G_1}{G_2}$ and match $u \colon U \to H$
  such that $f_1 = u \comp i_1$ and $f_2 = u \comp i_2$.
\end{lemma}
See \Fig{bimotor-mgs} in \Example{} for an example of a set of MGs.

\section{Graph-based GREs}
\label{sec:diff-sem}

To derive the GRE for graphs (GREG) we follow the development in our
previous work~\cite{icfem,rc15} with the important difference that we
do not assume derivations to be reversible.
The key technical innovation that allows us to avoid the assumption of
reversibility is the backward modularity lemma (\Lem{bmod}).

As sketched in \Sec{discussion}, our GRE for graphs is defined in
terms of graph observables, which we now define formally.
Fix $S$ to be the countable (up to iso) set of finite graphs, and let
$F \in S$ be a graph.
The \emph{graph observable} $\obs{F} \colon S \to \NN$ is the
integer-valued function $\obs{F}(G) := \nmchs{F}{G}$ counting the
number of occurrences (\ie matches) of $F$ in a given graph $G$.
Graph observables are elements of the vector space $\RR^S$ of
real-valued functions on $S$.

The stochastic rate matrix $Q_\al$ for a rule $\al\colon L \rto R$
defined in~\eqref{eq:rates}
is a linear map on~$\RR^S$.
Its action on an observable $\obs{F}$ is given by
\begin{align}
  (Q_\al \obs{F})(G) &:= \tsum_H q^\al_{G H}(\obs{F}(G) - \obs{F}(H))
  &\text{for } G, H \in S.
  \label{eq:jump}
\end{align}
Since the sum above is finite, $Q_\al \obs{F}$ is indeed a
well-defined element of $\RR^S$.
We call $Q_\al \obs{F}$ the \emph{jump} of $\obs{F}$ relative to
$Q_\al$.
Intuitively, $(Q_\al \obs{F})(G)$ is the expected rate of change in
$\obs{F}$ given that the CTMC sits at $G$.

To obtain the GREG as sketched in \Sec{intro}, we want to express the
jump as a finite linear combination of graph observables.
We start by substituting the definition of $Q_\al$ in \eqref{eq:jump}.
\begin{align*}
  (Q_\al \obs{F}) (G) \;
  &= \; \tsum_{H} q^\al_{G H} (\obs{F}(H) - \obs{F}(G)) \\ \nonumber
  &= \; \tsum_{H}
     \tsum_{
f \in \mchs{L}{G} \text{ s.t. } \al(f) \in \mchs{R}{H}} \,
     (\nmchs{F}{H} - \nmchs{F}{G}) \\ \nonumber
  &= \; \tsum_{f \in \mchs{L}{G}}
     \left( \nmchs{F}{\cod(\al(f))} - \nmchs{F}{G} \right).
\end{align*}
where the simplification in the last step is justified by the fact
that $f$ and $\al$ uniquely determine $\al(f)$.
The last line suggests a decomposition of $Q_\al \obs{F}$ as
$Q_\al \obs{F} = Q^+_\al \obs{F} - Q^-_\al \obs{F}$, where $Q^+_\al$
produces new instances of $F$ while $Q^-_\al$ consumes existing ones.

By \Lem{mgs}, we can factor the action of the consumption term $Q^-_\al$ through the MGs
$\MGs{L}{F}$ of $L$ and $F$ to obtain
\[
  (Q^-_\al \obs{F})(G)
  \; = \;
\tsum_{f \in \mchs{L}{G}} \nmchs{F}{G}
  \; = \;
  \abs{\mchs{L}{G}} \cdot \abs{\mchs{F}{G}}
  \; = \;
  \sum_{\mu \in \MGs{L}{F}} \nmchs{\hat{\mu}}{G}.
\]
The resulting sum is a linear combination of a finite number of
graph observables, which is exactly what we are looking for.

Simplifying the production term requires a bit more work.
Applying the same factorisation \Lem{mgs}, we arrive at
\begin{align*}
  (Q^+_\al \obs{F})(G) \;
  &= \;
  \tsum_{f \in \mchs{L}{G}} \nmchs{F}{\hat{\al}(f)} \\
  &= \;
  \tsum_{f \in \mchs{L}{G}} \;
  \tsum_{(\mu_1, \mu_2) \in \MGs{R}{F}}
  \abs{\setof{g \in \mchs{\hat{\mu}}{\hat{\al}(f)}}{
    g \comp \mu_1 = \al(f)}}.
\end{align*}
where $\hat{\al}(f) = \cod(\al(f))$ denotes the codomain of the
comatch of $f$.
To simplify this expression further, we use the properties of
derivations introduced in \Sec{deriv-prop}.
First, we observe that $\mu_1$ must be derivable by $\al$ for
the set of $g$'s in the above expression to be nonempty.
\begin{lemma}\label{lem:restriction}
  Let $\al \colon L \rto R$ be a rule and $f \colon L \to G$,
  $g \colon R \to H$, $g_1 \colon R \to T$ matches such that
  $f \rwto[\al] g$, but $g_1$ is not derivable by~$\al$.
Then there is no match $g_2 \colon T \to H$ such that
  $g_2 \comp g_1 = g$.
\end{lemma}
\begin{proof}
  By the contrapositive of backward modularity.
Any such $g_2$ would induce, by \Lem{bmod}, a match
  $f_1\colon L \to S$ and a derivation $f_1 \rwto[\al] g_1$.
\qed
\end{proof}
We may therefore restrict the set $\MGs{R}{F}$ of right-hand MGs under
consideration to the subset
$\RMGs{\al}{F} := \setof{ (\mu_1, \mu_2) \in \MGs{R}{F} }{\exists h. \, h \rwto[\al] \mu_1 }$
of MGs with a first projection derivable by $\al$.
Next, we observe that the modularity \Lem[Lemmas]{fmod} and
\ref{lem:bmod} establish a \emph{one-to-one correspondence} between
the set of factorisations of the comatches $\al(f)$ (through the MGs
in $\RMGs{\al}{F}$) and a set of factorisations of the corresponding
matches $f$.
\begin{lemma}[correspondence of matches]\label{lem:correspondence}
  Let $\al$, $\ba$, $\gamma$, $f$, $f_1$, $g$, $g_1$ such that diagrams~\eqref{dg:corr1} and \eqref{dg:corr2} are
  derivations and $g_1$ is derivable by $\al$.
Then the set
$M_L = \setof{f_2 \in \mchs{S}{G}}{f_2 \comp f_1 = f}$ is in
one-to-one correspondence with the set
$M_R = \{g_2 \in \mchs{T}{H} \st g_2 \comp g_1 = g\}$.
\begin{diags}
    \begin{tikzcd}
      L \dar[swap]{f}\rar[rule]{\al} & R \dar{g}\\
      G \rar[rule]{\beta}            & H
    \end{tikzcd}~\eqnlabel{dg:corr1}\hspace{3em}
    \begin{tikzcd}
      L \dar[swap]{f_1} & R \dar{g_1}\lar[elur, swap]{\; \dg{\al}}\\
      S                 & T \lar[elur, swap]{\; \dg{\gamma}}
    \end{tikzcd}~\eqnlabel{dg:corr2}
  \end{diags}
\end{lemma}
\begin{proof}Since $g_1$ is $\al$-derivable, the diagram~\eqref{dg:corr2} is
  reversible, that is, $f_1 \rwto[\al] g_1$, with corule $\gamma$
  (by \Lem{derivability}).
Hence, if we are given a match $f_2$ in $M_L$, we can
  forward-decompose \eqref{dg:corr1} vertically along the
  factorisation $f_2 \comp f_1 = f$, resulting in the diagram below
  (by forward modularity, \Lem{fmod}).
\begin{narrow}{8em}
    Furthermore, the comatch $g_2$ is unique with respect to this
    decomposition, thus defining a function $\phi \colon M_L \to M_R$
    that maps any $f_2$ in $M_L$ to the corresponding comatch
    $\phi(f_2) = g_2$ in $M_R$.
We want to show that $\phi$ is a bijection.  By backward
    modularity (\Lem{bmod}), there is a match $f_2 \in M_L$ for any
    match $g_2 \in M_R$ such that $\phi(f_2) = g_2$ (surjectivity), and
    furthermore, $f_2$ is the unique match for which $\phi(f_2) = g_2$
    (injectivity).
\qed
  \end{narrow}
\hfill\rightoverlay{10em}{\begin{tikzcd}[ampersand replacement=\&]
      L \dar[swap]{f_1}\rar[rule]{\al} \&
      R \dar[swap]{g_1}\arrow[bend left=40]{dd}{g}\\
      S \dar[swap]{f_2}\rar[rule]{\gamma} \& T \dar[swap, dashed]{g_2}\\
      G \rar[rule]{\beta} \& H
    \end{tikzcd}}
\vspace{-\baselineskip}
\end{proof}
Using \Lem[Lemmas]{restriction} and \ref{lem:correspondence},
we can simplify $Q^+_\al$ as follows:
\begin{align*}
  (Q^+_\al \obs{F})(G) \;
  &= \sum_{f \in \mchs{L}{G}} \;
     \sum_{\mu \in \RMGs{\al}{F}}
     \abs{\setof{g_2 \in \, \mchs{\hat{\mu}}{\hat{\al}(f)}}{
       g_2 \comp \mu_1 = \al(f)}} \\
  &= \sum_{\mu \in \RMGs{\al}{F}} \, \sum_{f \in \mchs{L}{G}}
     \abs{\setof{f_2 \in \mchs{\dg{\hat{\al}}(\mu_1)}{G}}{
       f_2 \comp \dg{\al}(\mu_1) = f}} \\
  &= \sum_{\mu \in \RMGs{\al}{F}} \nmchs{\dg{\hat{\al}}(\mu_1)}{G}
\end{align*}
If we set $\LMGs{\al}{F} := \MGs{L}{F}$ to symmetrise notation, we obtain
\begin{align}
  Q_\al(\obs{F})
  &=
  \tsum_{\mu \in \RMGs{\al}{F}} \obs{\dg{\hat{\al}}(\mu_1)} -
  \tsum_{\mu \in \LMGs{\al}{F}} \obs{\hat{\mu}}
  \label{eq:reg}
\end{align}

Now, in general for a CTMC on a countable state space $S$, the Markov-generated and time-dependent probability $p$ on $S$ follows the master equation~\cite{Norris98mc,anderson2012continuous}:
$\ddt p^T = p^T Q$. Given
an abstract observable $f$ in $\RR^S$,  and writing $\exp_p(f) := p^T f$ for the expected
value of $f$ according to $p$,
we can then derive the formal\footnote{In the present paper, we elide the subtle issues of ensuring that the system of interest actually satisfies this equation.
See the work of Spieksma~\cite{spieksma} for the underlying mathematics or our previous work \cite{Tobiasfossacs},
which additionally considers computability of the solutions to arbitrary precision.} Kolmogorov equation for
$f$:
\[
  \ddt \exp_p(f)
  \; = \;
  \ddt p^T f
  \; = \;
  p^T Q f
  \; = \;
  \exp_p (Q f),
  \label{eq:RE1}
\]
giving us an equation for the rate of change of the mean of $f(X(t))$.
Following this general recipe gives us the GRE for graphs immediately
from~\eqref{eq:reg}.
\begin{align}
  \dotavg{(\obs{F})} = &- \sum_{\al \in \rset} k(\al) \!\!\! \sum_{\mu \in \LMGs{\al}{F}}
       \!\!\! \avg{\obs{\hat{\mu}}}
   + \sum_{\al \in \rset} k(\al) \!\!\! \sum_{\mu \in \RMGs{\al}{F}}
       \!\!\! \avg{\obs{\dg{\hat{\al}}(\mu_1)}}.
  \label{eq:GREG}
\end{align}
Remember that $\mu_1$ denotes the left injection of the MG
$\mu = (\mu_1, \mu_2)$ while $\hat{\mu}$ denotes its codomain, and
that $\dg{\hat{\al}}(f) = \cod(\dg{\al}(f))$.

Unsurprisingly, the derivation of \eqref{eq:GREG} was more technically
challenging than that of the GRE for reversible graph rewrite systems
(\cf~\cite[Theorem~2]{rc15}).
Yet the resulting GREs look almost identical (\cf~\cite[Eq.~(7)]{rc15}).
The crucial difference is in the production term $Q^+_\al$, where we
no longer sum over the full set of right-hand MGs $\MGs{R}{F}$ but
only over the subset $\RMGs{\al}{F}$ of MGs that are $\al$-derivable.
This extra condition is the price we pay for dealing with
irreversibility:
irreversible rules can consume all MGs, but only produce some.

Note that the number of terms in \eqref{eq:GREG} depends on the size
of the relevant sets of left and right-hand MGs, which is worst-case
exponential in the size of the graphs involved, due to the
combinatorial nature of MGs.
(See \Fig{bimotor-mgs} in \Example{} for an example.)
In practice, one often finds many pairs of \emph{irrelevant} MGs, the
terms of which cancel out exactly.
This reduces the effective size of the equations but not the overall
complexity of generating the GREG.

Finally, as said in \Sec{discussion}, the repeated application of
\eqref{eq:GREG} will lead to an infinite expansion in general.
In practice, the system of ODEs needs to be truncated.
For certain models, one can identify invariants in the underlying
rewrite system via static analysis, which result in a finite closure
even though the set of reachable components is demonstrably
infinite~\cite{DanosHJS12sasb}.
We have seen an example in \Sec{intro}.

\section{Conclusion}
\label{sec:conclusions}

We have developed a computer-supported method for mean field approximations (MFA)
for stochastic systems with graph-like states that are described by rules of SPO rewriting.
The underlying theory unifies a large and seemingly unstructured collection of MFA approaches
which share a graphical ``air de famille''.
Based on the categorical frameworks of graph
transformation systems (GTS), we have developed MFA-specific techniques,
in particular concerning the combinatorics of minimal gluings.
The main technical hurdle consisted in showing that
the set of subgraph observables is closed under the action of the rate
matrix (a.k.a.\@ the infinitesimal generator) of the continuous-time Markov
chain generated by an irreversible GTS. The proof is constructive and gives us an
explicit term for the derivative of the mean of any observable of interest.

Mean field approximation and moment-closure methods are of wide use in applications,
as typical probabilistic systems tend to have state spaces which defy
more direct approaches. To reach their full potential, MFAs need
to be combined with reachability and invariant analysis
(as illustrated in \Example).

We have worked the construction at the general axiomatic level of
SPO-rewriting with matches and rules restricted to monomorphisms. One
interesting extension is to include {nested application conditions}
(NACs)~\cite{HabelP09mscs,Rensink04gg} where the
application of a rule can be modulated locally by the context of the
match. NACs are useful in practice, and bring aboard the expressive
power of first order logic in the description of transformation
rules. We plan to investigate the extension of our approach to NACs,
and, in particular, whether it is possible to incorporate them
axiomatically, and what additional complexity cost they might incur.

Another direction of future work is to improve on the method of truncation.
In the literature, one often finds graphical MFAs used in combination
with conditional independence assumptions to control the size of connected
observables, as \eg the so-called pair approximation~\cite{durrett2012graph,gleeson}. As these methods are known to improve the accuracy of naive truncation,
we wish to understand if and how they can be brought inside our formal approach.

\bibliographystyle{splncs04}

\appendix

\section{Pushout and pull-back complements}
\label{app:defs}

Algebraic graph rewriting relies on certain
category-theoretical \emph{limits} and \emph{colimits}~\cite{BarrW95catcs}.  We give
definitions of the relevant (co-)\allowbreak limits here along with
some of their basic properties.  Among these, pullback
complements are the least known. We refer the
interested reader to Ref.~\cite{DyckhoffT87jpaa,CorradiniHHK06icgt} for a
thorough treatment.
\begin{diags}
  \begin{tikzcd}
    Q \arrow[bend right]{ddr}[swap]{g_2}\arrow[bend left]{drr}{g_1}
    \drar[dashed]{u} &[-1ex]\\[-1ex]
    & P \rar{p_1}\dar[swap]{p_2} & Y \dar{f_2}\\
    & X \rar[swap]{f_1} & Z
  \end{tikzcd}~\eqnlabel{dg:pb}\hspace{2ex}
  \begin{tikzcd}
    Z \rar{f_1}\dar[swap]{f_2} & X \dar{i_2}\arrow[bend left]{ddr}{g_2}\\
    Y \rar[swap]{i_1}\arrow[bend right]{drr}[swap]{g_1} &
    P \drar[dashed]{u} &[-1ex]\\[-1ex]
    & & Q
  \end{tikzcd}~\eqnlabel{dg:po}
\end{diags}

Let $\cat{C}$ be a category.
\begin{definition}[Pullback]\label{def:pb}
  A \emph{pullback} of a cospan of morphisms $X \to[f_1] Z \gets[f_2]
  X$ in $\cat{C}$ is a span $X \gets[p_1] P \to[p_2] Y$
  making the bottom-right square in~\eqref{dg:pb} commute, and such that
  for any other span $X \gets[g_1] Q \to[g_2] Y$ for which the outer
  square commutes, there is a unique morphism $u \colon Q \to P$
  making the diagram commute.
\end{definition}

\begin{definition}[Pushout]\label{def:po}
  A \emph{pushout} of a span of morphisms $X \to[f_1] Z \gets[f_2] Y$
  in $\cat{C}$ is a cospan $X \to[i_1] P \gets[i_2] Y$
  making the top-left square in~\eqref{dg:po} commute, and such that for
  any other cospan $X \to[g_1] Q \gets[g_2] Y$ for which the outer
  square commutes, there is a unique morphism $u \colon P \to Q$
  making the diagram commute.
\end{definition}

\begin{diags}
  \begin{tikzcd}
    P\arrow[bend left=20, near end]{drr}{f_1'}\arrow{ddd}[swap]{g_1'}
    \drar[swap]{p}\\[-3ex]
    & X \rar[swap]{f_1}\dar[swap]{g_1} & Y \dar{f_2}\\
    & W \rar{g_2} & Z\\[-3ex]
    Q \urar[dashed]{u}\arrow[bend right=20, near end]{urr}[swap]{g_2'}
  \end{tikzcd}~\eqnlabel{dg:fpbc}
\end{diags}

\begin{definition}[Final pullback complement]\label{def:fpbc}
  A \emph{final pullback complement (FPBC)} (or simply \emph{pullback
    complement}) of a pair of composable morphisms $X \to[f_1] Y
  \to[f_2] Z$ in some category $\cat{C}$ is a pair of composable
  morphisms $X \to[g_1] W \to[g_2] Z$ making the right inner square
  in~\eqref{dg:fpbc} a pullback, such that for any other pullback $P
  \to[f_1'] Y \to[f_2] Z \gets[g_2'] Q \gets[g_1'] P$ and morphism $p
  \colon P \to X$ for which the diagram commutes, there is a unique
  morphism $u \colon Q \to W$ that makes the diagram commute.
\end{definition}

The following lemmas, pertaining to the composition of pullbacks,
pushouts and FPBCs, respectively, are used throughout the proofs in
\App{proofs}.  The first two are dual versions of the well-known
``pasting'' lemma for pullbacks and pushouts, and we leave their
proofs as an exercise to the reader.  A proof of the third lemma can
be found in~\cite[Proposition~5]{Lowe10gg}.
\begin{diags}
  \begin{tikzcd}[row sep=5ex, column sep=6ex]
    A \dar[swap]{g_1}\rar{f_1} & B \dar{g_2}\rar{f_2} & C \dar{g_3}\\
    D \rar[swap]{h_1} & E \rar[swap]{h_2} & F
  \end{tikzcd}\hspace{2em}\eqnlabel{dg:pasting}
\end{diags}
\begin{lemma}[Pasting of pullbacks]
  Suppose the right inner square in~\eqref{dg:pasting} is a pullback
  in some category~$\cat{C}$.  Then the left inner square is a
  pullback if and only if the outer square is.
\end{lemma}
\begin{lemma}[Pasting of pushouts]
  Suppose the left inner square in~\eqref{dg:pasting} is a pushout in
  some category~$\cat{C}$.  Then the right inner square is a pushout
  if and only if the outer square is.
\end{lemma}
\begin{lemma}[Composition of FPBCs]
  Consider again diagram~\eqref{dg:pasting} in some
  category~$\cat{C}$,
  \begin{itemize}
  \item (horizontal composition) if $A \to[g_1] D \to[h_1] E$ and $B
    \to[g_2] E \to[h_2] F$ are the FPBCs of $A \to[f_1] B \to[g_2] E$
    and $B \to[f_2] C \to[g_3] F$, respectively, then $A \to[g_1] D
    \to[h_2 \comp h_1] F$ is the FPBC of $A \to[f_2 \comp f_1] C
    \to[g_3] F$;
  \item (vertical composition) if $A \to[f_1] B \to[g_2] E$ and $B
    \to[f_2] C \to[g_3] F$ are the FPBCs of $A \to[g_1] D \to[h_1] E$
    and $B \to[g_2] E \to[h_2] F$, respectively, then $A \to[f_2 \comp
    f_1] C \to[g_3] F$ is the FPBC of $A \to[g_1] D \to[h_2 \comp h_1]
    F$.
  \end{itemize}
\end{lemma}

\section{Generalised proofs of lemmas}
\label{app:proofs}

This section contains detailed proofs of the various lemmas introduced
in previous sections.  We will present the proofs in a slightly more
general setting, namely that of \emph{sesqui-pushout (SqPO)
  rewriting}~\cite{CorradiniHHK06icgt} in arbitrary \emph{adhesive
  categories}~\cite{LackS05ita}.  To be precise, we assume an ambient
category $\Gs$, such that
\begin{itemize}
\item $\Gs$ is adhesive (among other things, this implies that $\Gs$
  has all pullbacks as well as all pushouts along monomorphisms, that
  monomorphism are stable under pushout, and that all such pushouts
  are also pullbacks, cf.~\cite{LackS05ita}),
\item $\Gs$ has all final pullback complements (FPBCs) above
  monomorphisms.
\end{itemize}
Both these assumptions hold in $\Grph$.  Within $\Gs$, we define
derivations as in \Def{deriv}, taking matches and rules to be
monomorphisms and spans thereof, respectively.

Alternatively, rules can be seen as partial maps~\cite{RobinsonR88ic}
in the category $\GsP$, generalising the interpretation of rules as
partial graph morphisms in $\GrphP$.  Derivations can then be shown to
correspond exactly to pushouts of rules along monomorphisms in
$\GsP$~\cite[Proposition 2.10]{BaldanCHKS14mscs}, and composition of
derivations corresponds to pushout composition in $\GsP$.

\subsection{Proof of \Lem{mgs} (minimal gluings)}
\label{app:mgs}

Let $G_1$ and $G_2$ be graphs, then
\begin{enumerate}
\item the set $\MGs{G_1}{G_2}$ of MGs of $G_1$ and $G_1$ is
  finite, and
\item for every cospan $G_1 \to[f_1] H \gets[f_2] G_2$ of matches,
  there is a unique MG $(G_1 \to[i_1] U \gets[i_2] G_2)
  \in \MGs{G_1}{G_2}$ and match $u \colon U \to H$ such that $f_1 = u
  \comp i_1$ and $f_2 = u \comp i_2$.
\end{enumerate}

\begin{proof}\label{prf:mgs}
  For this proof we will make two additional assumptions on $\Gs$,
  namely that $\Gs$ has all binary products, and that the objects of
  $\Gs$ are finitely powered, that is, any object $A$ in $\Gs$ has a
  finite number of subobjects.  Both these assumptions hold in
  $\Grph$.

Recall that the subobjects of any object $A$ in $\Gs$ form a poset
category $\Sub(A)$ with \emph{subobject intersections} as products and
\emph{subobject unions} as coproducts.  By stability of monomorphisms
under pullback, products (intersections) in $\Sub(A)$ are given by
pullbacks in $\Gs$, and since $\Gs$ is adhesive, coproducts (unions)
in $\Sub(A)$ are given by pushouts of pullbacks in $\Gs$.
See~\cite[Theorem~5.1]{LackS05ita} for more details.
\begin{diags}
  \begin{tikzcd}[column sep=2ex, row sep=2.5ex]
    {}& G_1 \cap G_2 \dlar[swap]{p_1}\drar{p_2}\\
    G_1 \drar[swap]{i_1}\arrow[bend right]{ddr}[swap]{f_1} & &
    G_2 \dlar{i_2}\arrow[bend left]{ddl}{f_2}\\
    & G_1 \cup G_2 \dar[dashed]{u}\\[1ex]
    & H
  \end{tikzcd}\hspace{2em}\eqnlabel{dg:mgs1}
\end{diags}
We will start by showing that any cospan $G_1 \to[f_1] H \gets[f_2]
G_2$ of matches in $\Gs$ factorises uniquely through an element of
$\MGs{G_1}{G_2}$.  Given such a cospan, let $u \colon G_1 \cup G_2 \to
H$ be a representative in $\Gs$ of the subject union of $f_1$ and
$f_2$ in $\Sub(H)$, with coproduct injections $i_1 \colon G_1 \to G_1
\cup G_2$ and $i_2 \colon G_2 \to G_1 \cup G_2$ as in~\eqref{dg:mgs1}.
Since $u$ is the mediating morphism of a pullback, it is unique up to
isomorphism of $G_1 \cup G_2$.  It remains to show that $G_1 \to[i_1]
G_1 \cup G_2 \gets[i_2] G_2$ is a MG.  By adhesiveness of
$\Gs$, the pushout square at the top of~\eqref{dg:mgs1} is also a
pullback, and hence an intersection of $i_1$ and $i_2$ in $\Sub(G_1
\cup G_2)$.  It follows that $\id_{G_1 \cup G_2}$ represents the
subobject union of $i_2$ and $i_2$ in $\Sub(G_1 \cup G_2)$ and hence
$G_1 \to[i_1] G_1 \cup G_2 \gets[i_2] G_2$ is indeed a MG.

The finiteness of $\MGs{G_1}{G_2}$ follows from a similar argument.
First, note that $\abs{\MGs{G_1}{G_2}} = \abs{\MGCls{G_1}{G_2}}$,
so it is sufficient to show that $\MGCls{G_1}{G_2}$ is finite.
Being a subobject union, every MG is the pushout of a span
$G_1 \gets[p_1] G_1 \cap G_2 \to[p_2] G_2$ of matches as
in~\eqref{dg:mgs1}.  Since isomorphic spans have isomorphic pushouts,
there can be at most as many isomorphism classes of MGs of
$G_1$ and $G_2$ as there are isomorphism classes of spans over $G_1$
and $G_2$.  Furthermore, the spans $G_1 \gets[p_1] X \to[p_2] G_2$ are
in one-to-one correspondence with the pairings $\pair{p_1}{p_2} \colon
X \to G_1 \times G_2$ in $\Gs$, which represent subobjects in
$\Sub(G_1 \times G_2)$ (with isomorphic spans corresponding to
identical subobjects).  Since $G_1 \times G_2$ is finitely powered,
there are only a finite number of such subobjects, and hence there can
only be a finite number of isomorphism classes of spans over $G_1$ and
$G_2$, which concludes the proof.
\qed
\end{proof}

\subsection{Proof of \Lem{fmod} (forward modularity)}
\label{app:fmod}

Let $f_1$, $f_2$, $g$, $g_1$ be matches, and $\al$, $\ba$, $\gamma$
rules, such that the diagrams~\eqref{dg:fmod1} and \eqref{dg:fmod2}
are derivations.  Then there is a unique match $g_2$, such that
diagram~\eqref{dg:fmod3} commutes and is a vertical composition of
derivations.
\begin{diags}
  \begin{tikzcd}
    L \dar[swap]{f_1}\rar[rule]{\al} & R \arrow{dd}{g}\\
    S \dar[swap]{f_2}\\
    G \rar[rule]{\beta}           & H
  \end{tikzcd}~\eqref{dg:fmod1}\hspace{1em}
  \begin{tikzcd}
    L \dar[swap]{f_1}\rar[rule]{\al} & L \dar{g_1}\\
    S \rar[rule]{\gamma}                & T
  \end{tikzcd}~\eqref{dg:fmod2}\hspace{1em}
  \begin{tikzcd}
    L \dar[swap]{f_1}\rar[rule]{\al} &
    R \dar[swap]{g_1}\arrow[bend left=40]{dd}{g}\\
    S \dar[swap]{f_2}\rar[rule]{\gamma} & T \dar[swap, dashed]{g_2}\\
    G \rar[rule]{\beta}           & H
  \end{tikzcd}~\eqref{dg:fmod3}
\end{diags}

\begin{proof}\label{prf:fmod}
  Using the universal property of the pushout~\eqref{dg:fmod2}, we
  obtain the mediating morphism $g_2$ and apply the Pasting Lemma for
  pushouts to conclude that the lower square in~\eqref{dg:fmod3} is a
  pushout.
\qed
\end{proof}

\subsection{Proof of \Lem{bmod} (backward modularity)}
\label{app:bmod}

Let $f$, $f_1$, $g_1$, $g_2$ be matches, and $\al$, $\ba$, $\gamma$
rules, such that the diagrams~\eqref{dg:bmod1} and \eqref{dg:bmod2}
are derivations.  Then there is a unique match $f_2$, such that
diagram~\eqref{dg:bmod3} commutes and is a vertical composition of
derivations.
\begin{diags}
  \begin{tikzcd}
    L \arrow{dd}[swap]{f}\rar[rule]{\al} & R \dar{g_1}\\
    & T \dar{g_2}\\
    G \rar[rule]{\beta}                     & H
  \end{tikzcd}~\eqref{dg:bmod1}\hspace{1em}
  \begin{tikzcd}
    L \dar[swap]{f_1} & R \dar{g_1}\lar[elur, swap]{\dg{\al}}\\
    S                 & T \lar[elur, swap]{\dg{\gamma}}
  \end{tikzcd}~\eqref{dg:bmod2}\hspace{1em}
  \begin{tikzcd}
    L \arrow[bend right=40]{dd}[swap]{f}\rar[rule]{\al}\dar{f_1} &
    R \dar{g_1}\\
    S \dar[dashed]{f_2}\rar[rule]{\gamma} & T \dar{g_2}\\
    G \rar[rule]{\beta}                   & H
  \end{tikzcd}~\eqref{dg:bmod3}
\end{diags}

\begin{proof}\label{prf:bmod}
The proof is in three steps: we first construct $f_1$ and $f_2$ in
$\Gs$, then we show that diagram~\eqref{dg:bmod3} is indeed a
composition of derivations, and finally we verify the uniqueness of
$f_2$ for this property.

Consider diagram~\eqref{dg:bmod4} below, which is the underlying
diagram in $\Gs$ of derivation~\eqref{dg:bmod1} from the lemma:
\begin{diags}
  \begin{tikzcd}
    L \arrow{dd}[swap]{f} & K \lar[swap]{\al_1}
    \arrow{dd}[swap]{h}\rar{\al_2} & R \dar{g_1}\\
    & & T \dar{g_2}\\
    G & D \lar[swap]{\beta_1}\rar{\beta_2} & H
  \end{tikzcd}~\eqnlabel{dg:bmod4}\hspace{2em}
  \begin{tikzcd}
    L \arrow[bend right=40]{dd}[swap]{f}\dar{f_1} &
    K \lar[swap]{\al_1}\rar{\al_2}\dar{h_1} & R \dar{g_1}\\
    S \dar{f_2} & E \lar[swap]{\gamma_1}\rar{\gamma_2}\dar{h_2} &
    T \dar{g_2}\\
    G & D \lar[swap]{\beta_1}\rar{\beta_2} & H
  \end{tikzcd}~\eqnlabel{dg:bmod5}
\end{diags}
The right-hand square is a pushout along monomorphisms, and hence it
is also a pullback in $\Gs$, and we can decompose it along $g_1$ and
$g_2$ to obtain the upper and lower right squares of
diagram~\eqref{dg:bmod5}.  By stability of pushouts in $\Gs$
(see~\cite[Lemma~4.7]{LackS05ita}), both these squares are also
pushouts.  To complete diagram~\eqref{dg:bmod5}, let its upper-left
square be a pushout, and $f_2$ the unique mediating morphism such that
the right-hand side of the diagram commutes.

Note that all morphisms in~\eqref{dg:bmod5}, except possibly $f_2$,
are monic.  The composites $\gamma_1 \comp h_1$ and $\gamma_2 \comp
h_1$ are pushout complements of $f_1 \comp \al_1$ and $g_1 \comp
\al_2$, respectively, and hence
by~\cite[proposition~12]{CorradiniHHK06icgt}, they are also FPBCs.  It
follow that the upper half of~\eqref{dg:bmod5} is indeed the
underlying diagram in $\Gs$ of both the derivations
in~\eqref{dg:bmod2} and the upper half of~\eqref{dg:bmod3}.  For the
lower half of~\eqref{dg:bmod5} to also be a derivation, $f_2$ must be
a match, so we need to show that it is monic.  To show this, let $S
\to[i_1] G' \gets[i_2] D$ be a pushout of $S \gets[\gamma_1] E
\to[h_2] D$, and let $u \colon G' \to G$ be its mediating morphism
with respect to the lower-left square of~\eqref{dg:bmod5}.  Since
$h_2$ is monic, so is $i_1$ (by adhesiveness of $\Gs$).  By pasting of
pushouts, $u$ is also the mediating morphism of the pushout $L \to[i_1
\comp f_1] G' \gets[i_2] D$ with respect to the left-hand square
in~\eqref{dg:bmod4}, which in turn, is also a pullback square.  In
fact, the composite pushout is the union of the subobjects represented
by $f$ and $\beta_1$, and hence by~\cite[Theorem~5.1]{LackS05ita},
$u$ is a monomorphism.  It then follows that the composite $f_2 = u
\comp i_1$ is also a monomorphism.
\begin{diags}
  \begin{tikzcd}
    L \arrow[bend right=40]{dd}[swap]{f}\dar{f_1} &
    K \lar[swap]{\al_1}\rar{\al_2}\dar{h_1} &
    R \dar[swap]{g_1}\arrow[bend left=40]{dd}{g}\\
    S \dar{f_2} &
    E \lar[swap]{\gamma_1}\rar{\gamma_2}\dar{h_2'} &
    T \dar[swap]{g_2'}\\
    G & D \lar[swap]{\beta_1}\rar{\beta_2} & H
  \end{tikzcd}~\eqnlabel{dg:bmod6}\hspace{.5em}
  \begin{tikzcd}
    L \arrow[bend right=40]{dd}[swap]{f}\dar{f_1} &
    K \lar[swap]{\al_1}\rar{\al_2}\dar{h_1} &
    R \dar[swap]{g_1}\arrow[bend left=40]{dd}{g}\\
    S \dar{f_2'} &
    E \lar[swap]{\gamma_1}\rar{\gamma_2}\dar{h_2''} &
    T \dar[swap]{g_2}\\
    G & D \lar[swap]{\beta_1}\rar{\beta_2} & H
  \end{tikzcd}~\eqnlabel{dg:bmod7}
\end{diags}
Now let $h_2'$ and $g_2'$ be matches such that
diagram~\eqref{dg:bmod6} commutes and is a composition of tiles as per
\Lem{fmod}.  Then we have $\beta_1 \comp h_2 = f_2 \comp
\gamma_1 = \beta_1 \comp h_2'$, and hence $h_2 = h_2'$ because
$\beta_1$ is monic.  Furthermore, the top-right square
of~\eqref{dg:bmod6} is a pushout, and hence $g_2'$ is the unique
mediating morphism such that $\beta_2 \comp h_2 = g_2' \comp \gamma_2$
and $g = g_2' \comp g_1$.  But from diagram~\eqref{dg:bmod5} we know
that $\beta_2 \comp h_2 = g_2 \comp \gamma_2$ and $g = g_2 \comp g_1$,
and hence $g_2' = g_2$.  It follows that the bottom half
of~\eqref{dg:bmod5} is indeed a derivation.

Finally, let $f_2'$ and $h_2''$ be any matches such that
diagram~\eqref{dg:bmod7} commutes and is a composition of tiles.  Then
$h_2'' = h_2$ (because
$\beta_2 \comp h_2'' = g_2 \comp \gamma_2 = \beta_2 \comp h_2$ and
$\beta_2$ is monic) and $f_2' = f_2$ (because it is the unique
mediating morphism of the top-left pushout-square such that
$\beta_1 \comp h_2 = f_2' \comp \gamma_1$ and $f = f_2' \comp f_1$),
which concludes the proof.
\qed
\end{proof}

\subsection{Proof of \Lem{derivability} (derivability)}
\label{app:derivability}

A match $g \colon R \to H$ is derivable by a rule $\al \colon L
\rto R$ if and only if $g \rwto[\al \comp \dg{\al}] g$.
Equivalently, $g$ is derivable from $f$ by $\al$ if and only if the
derivation $g \rwto[\dg{\al}] f$ is reversible.

\begin{proof}\label{prf:derivability}
  This is a direct consequence of \Lem{bmod}.  First, assume that
  $g\colon R \to H$ is derivable by $\al\colon L \rto R$ from some
  match $f\colon L \to G$, and let $h\colon L \to E$ be the comatch of
  some derivation $g \rwto[\dg{\al}] h$.  By \Lem{bmod} (setting $g_1
  = g$ and $f_1 = h$), the derivation $h \rwto[\al] g$ exists, and so
  does $g \rwto[\al \comp \dg{\al}] g$ (by horizontal composition of
  derivations).

  Now assume that we are given the derivation $g \rwto[\al \comp
  \dg{\al}] g$ instead, and let $f'\colon L \to G'$ and $h'\colon R
  \to E'$ be the comatches of some derivations $g \rwto[\dg{\al}] f'$
  and $f' \rwto[\al] h'$.  By horizontal composition and uniqueness of
  derivations up to isomorphism, we have $g \rwto[\al \comp \dg{\al}]
  h'$ and $g = u \comp h'$ for some (unique) isomorphism $u\colon E'
  \toby{\simeq} H$.  Hence there is a derivation $f' \rwto[\al] g$.
\qed
\end{proof}

\end{document}